\documentclass{llncs}

\usepackage[english]{babel}
\usepackage[hang,centerlast]{subfigure}
\usepackage[utf8]{inputenc}
\usepackage{amsmath,amssymb,multirow,ifthen,enumerate,datetime}

\usepackage{gastex}
\usepackage{tikz}
\usetikzlibrary{automata,patterns,topaths,shapes,calc,positioning,fit}
\usetikzlibrary{decorations.pathreplacing}
\tikzstyle{every picture}+=[>=stealth,initial text={}]
\tikzstyle{accepting}=[accepting by arrow]

\title{Interrupt Timed Automata with Auxiliary Clocks and Parameters\thanks{This work has been supported by project ImpRo ANR-2010-BLAN-0317}}

\author{B. B\'erard\inst{1} \and S. Haddad\inst{2} \and A. Jovanovi\'c \inst{3} \and D. Lime\inst{4}}

\institute{
    Sorbonne Universit\'e, UPMC-Paris 6, CNRS UMR 7606, Paris, France\\
    \email{Beatrice.Berard{@}lip6.fr}\\
    \and 
    ENS Cachan, LSV, CNRS, INRIA, Cachan, France\\
     \email{haddad{@}lsv.ens-cachan.fr}
    \and 
    Department of Computer Science, University of Oxford, Oxford, UK\\
    \email{Aleksandra.Jovanovic{@}cs.ox.ac.uk}
    \and 
    \'Ecole Centrale de Nantes, IRCCyN, CNRS, Nantes, France\\
    \email{Didier.Lime{@}ec-nantes.fr}
}

\newcommand{\rel}{\bowtie}
\newcommand{\N}{\ensuremath{\mathbb{N}}}

\newcommand{\Q}{\ensuremath{\mathbb{Q}}}
\newcommand{\R}{\ensuremath{\mathbb{R}}}

\newcommand{\A}{\mathcal{A}}

\newcommand{\C}{\mathcal{C}}

\newcommand{\T}{\mathcal{T}}

\newcommand{\Rl}{\mathcal{R}}

\newcommand{\La}{\mathcal{L}}
\newcommand{\Li}{\mathcal{L}in}
\newcommand{\Diff}{\mathcal{D}{\it iff}}
\newcommand{\Po}{\mathcal{P}ol}
\newcommand{\Fr}{\mathcal{F}rac}
\newcommand{\U}{\mathcal{U}}
\newcommand{\tr}{\xrightarrow}
\newcommand{\fee}{\varphi}
\newcommand{\eps}{\varepsilon}
\newcommand{\vect}[1]{\mathbf{#1}}

\newcommand{\act}{{\tt act}}

\def\sg{\mathrel[\joinrel\mathrel[} 
\def\sd{\mathrel]\joinrel\mathrel]} 
\newcommand{\sem}[1]{\sg\!\!#1\!\!\sd}

\newcommand{\timedtrans}[3]{%
  \begin{tabular}{c}
    \(#1\) \\ 
    \(#2\) \\ 
    \(#3\) 
  \end{tabular}
}

\begin{document}

\maketitle

\begin{abstract}
  Interrupt Timed Automata (ITA), an expressive timed model, has been
  introduced in order to take into account interruptions, according to
  levels.  Due to this feature, this formalism is incomparable with
  Timed Automata.  However several decidability results related to
  reachability and model checking have been obtained.  We add
  auxiliary clocks to ITA, thereby extending its expressive power
  while preserving decidability of reachability.  Moreover, we define
  a parametrized version of ITA, with polynomials of parameters
  appearing in guards and updates.  While parametric reasoning is
  particularly relevant for timed models, it very often leads to
  undecidability results. We prove that various reachability problems,
  including \emph{robust reachability}, are decidable for this model,
  and we give complexity upper bounds for a fixed or variable number
  of clocks, levels and parameters.
\end{abstract}

\section{Introduction}\label{sec:intro}

\paragraph{Timed and hybrid models.} In order to model timed systems, the
expressive model of Hybrid Automata (HA) has been
proposed~\cite{alur95}.  Since its expressive power leads to the
undecidability of most verification problems, several semi-decision
procedures have been designed fo HA as well as subclasses with
decidability results like Timed Automata (TA)~\cite{alur90}.  The
model of interrupt timed automata (ITA)~\cite{berard09,berard12} was
proposed as a subclass of hybrid automata, incomparable with the class
of timed automata, where task interruptions are taken into
account. Hence ITA are particularly suited for the modelling of
scheduling with preemption.

\paragraph{Parametric verification.} Getting a complete knowledge of a
system is often impossible, especially when integrating quantitative
constraints. Moreover, even if these constraints are known, when the
execution of the system slightly deviates from the expected behaviour,
due to implementation choices, previously established properties may
not hold anymore. Additionally, considering a wide range of values for
constants allows for a more flexible and robust design.  

Introducing parameters instead of concrete values is an elegant way of
addressing these three issues.  Parametrisation however makes
verification more difficult. Besides, it raises new problems like
parameter synthesis, \textit{i.e.}, finding the set (or a subset) of
values for which some property holds.

\paragraph{Parameters for timed models.} Among quantitative
features, parametric reasoning is particularly relevant for timing
requirements, like network delays, time-outs, response times or clock
drifts.

Pioneering work on parametric real time reasoning was presented
in~\cite{alur93} for the now classical model of timed automata, with
parameter expressions replacing the constants to be compared with
clock values. Since then, many studies have been devoted to the
parametric verification of timed
models~\cite{berard99,miller00,doyen07}, mostly establishing
undecidability results for questions like parametric reachability,
even for a small number of clocks or parameters.  Relaxing
completeness requirement or guaranteed termination, several methods
and tools have been developed for parameter synthesis in timed
automata~\cite{andre09,andre12,jlr13}, as well as in hybrid
automata~\cite{alur96b,henzinger97}.  Another research direction
consists in defining subclasses of parametric timed models for which
some problems become
decidable~\cite{bozzelli-FMSD-09,hune-jlap-02,jflr12}. Unfortunately,
these subclasses are severely restricted.  It is then a challenging
issue to define expressive parametric timed models where reachability
problems are decidable.

\paragraph{Contributions.} 
Our contributions are twofold. First we define a more expressive version of ITA,
including auxiliary clocks. We prove that this new model is strictly more
expressive than the former one but retains decidability for the
reachability problem. With respect to the complexity issues,
we provide upper bounds: 2-EXPTIME in the general case,
PSPACE when the number of levels is fixed and PTIME when the number of clocks
is fixed.
We also give a PSPACE matching lower bound when the number of levels is fixed.

Our second contribution is to enrich ITA with parameters in the spirit
above. A PITA is a parametric version of ITA where polynomial
parameter expressions can be combined with clock values both as
additive and multiplicative coefficients.  Considering only additive
parametrisation, we reduce reachability to the same problem in basic
ITA. This reduction entails complexity upper bounds of respectively
2-EXPTIME, PSPACE when the number of levels is fixed and PTIME when
the number of clocks and parameters is fixed.  The multiplicative
setting is much more expressive and also very useful in practice, for
instance to model clock drifts. We prove that reachability in
parametric ITA is decidable as well as its robust variant, an
important property for implementation issues.  To the best of our
knowledge, this is the first time such a result has been obtained for
a model including a multiplicative parametrization.  Furthermore, we
establish upper bounds for the computational complexity: 2-EXPSPACE
and PSPACE when the number of levels is fixed.  Our technique combines
the construction of symbolic class automata from the unparametrized
case and the first order theory of real numbers.

\paragraph{Outline.} The model of Interrupt Timed Automata with
auxiliary clocks is defined in Section~\ref{sec:modele}, with
reachability analysis in Section~\ref{sec:reach-ita}. The parametric
ITA model is introduced in Section~\ref{sec:pmodele}. The reachability
analysis is split into two sections: the additive case is handled in
Section~\ref{sec:amodele} while the results for the multiplicative
case are given in Section~\ref{sec:multpar}. We conclude and give some
perpectives for this work in Section~\ref{sec:conc}.

\section{Interrupt Timed Automata}\label{sec:modele}

\subsection{Notations}
The sets of natural, rational and real numbers are denoted
respectively by $\N$, $\Q$ and $\R$. Given an alphabet $\Sigma$, we
denote by $\Sigma^*$ the set of finite words over $\Sigma$, with
$\eps$ the empty word. The set of timed words over $\Sigma$ is the set
of finite sequences of the form $(a_1,t_1) \ldots (a_n,t_n)$ where
$a_i \in \Sigma$ for all $i \in \{1, \ldots,n\}$ and $(t_i)_{1\leq
  i\leq n}$ is a non decreasing sequence of real numbers. A timed
language is a set of timed words. For a timed word $w=(a_1,t_1) \ldots
(a_n,t_n)$, we define $Untime(w)=a_1 \ldots a_n$ as its projection on
$\Sigma^*$ and for a timed language $L$, we set $Untime(L) =
\{Untime(w) \mid w\in L\}$.

Given two sets $F,G$ with $F$ finite, we denote by $\Li(F,G)$ the set
of linear expressions $\sum_{f\in F} a_ff+b$ where the $a_f$'s and $b$
belong to $G$. We also denote by $\Diff(F)$ the set of expressions
$f-f'$ with $f,f' \in F$.

\paragraph{Clock constraints.} Let $X$ be a finite set of clocks and
let $Y,Z$ be disjoint subsets of $X$. We denote by $\C(Y,Z)$ the set of
constraints obtained by conjunctions of atomic propositions of the
form $C \rel 0$, where $C$ is an expression in $\bigcup_{y \in Y}
\Li(Z\cup\{y\},\Q) \cup \Diff(Y)$ and $\rel \,\in
\{>,\geq,=,\leq,<\}$. Such a constraint either compares with
zero a linear expression of clocks in $Y\cup Z$ including at most
one clock of $Y$, or compares two clocks of $Y$.  We also set
$\C(X)=\bigcup_{Y,Z \subseteq X} \C(Y,Z)$.

\paragraph{Updates.} An \emph{update} over $X$ is a conjunction of
assignments of the form $\wedge_{y \in Y}\ y:= C_y$, where $Y
\subseteq X$ and $C_y \in \Li(X,\Q)$. The set of updates is written
$\U(X)$.  For an expression $C$ and an update $u$, the expression
$C[u]$ is obtained by ``applying'' $u$ to $C$, \textit{i.e.},
simultaneously substituting each $x$ by $C_x$ in $C$, if $x := C_x$ is
the update for $x$ in $u$.  For instance, for clocks $X = \{x_1,
x_2\}$, expression $C= 2x_2 -2x_1 + 3$ and the update $u$ defined by
$x_1 := 1 \wedge x_2 := 3x_1 +2$, applying $u$ to $C$ yields the
expression $C[u]= 2(3x_1 +2) -2(1) + 3 =6x_1 + 5$.

\paragraph{Valuations.} A \emph{clock valuation} is a mapping $v : X
\mapsto \R$, with $\vect{0}$ the valuation where all clocks have value
$0$.  For a valuation $v$ and an expression $C \in \Li(X,\Q)$, we note
$v(C) \in \R$ the result of evaluating $C$ w.r.t. $v$.  Given an
update $u$ and a valuation $v$, the valuation $v[u]$ is defined by
$v[u](x)=v(x)$ if $x$ is unchanged by $u$ and $v[u](x) = v(C_x)$ if $x
:= C_x$ is the update for $x$ in $u$.  For instance, let $X = \{x_1,
x_2, x_3\}$ be a set of three clocks. For valuation $v = (2, 1.5, 3)$
and update $u$ defined by $x_1 := 1 \wedge x_3 := x_3 - x_1$, applying
$u$ to $v$ yields the valuation $v[u] = (1, 1.5, 1)$.

\subsection{Interrupt Timed Automata}

\paragraph{Definitions.} The behaviour of an ITA can be viewed as the
one of an operating system with interrupt levels.  With each level
are associated a set of states and a set of clocks partitionned into a
\emph{main} clock and \emph{auxiliary} clocks.  In a state of a given
level, exactly one clock of this level is active (rate $1$), while the
other clocks at lower or equal levels are suspended (rate $0$), and
the clocks at higher levels are not yet activated and thus contain
value $0$.  The enabling conditions on transitions, called
\emph{guards}, are constraints over clocks of the current level or
main clocks of lower levels (with some restrictions).  Transitions can
\emph{update} the clock values.  If the transition decreases
(resp. increases) the level, then each clock which is relevant after
(resp. before) the transition can (1) be left unchanged, (2) be
updated with a linear expression of main clocks of strictly lower
levels or (3) be updated with another clock at the same level (with
some restrictions).  Roughly speaking, the restrictions are introduced
to forbid at some level any (direct or indirect) influence of the
auxiliary clocks at lower levels on the behaviour of the ITA.

\begin{definition}
\label{def:ita}
  An \emph{interrupt timed automaton} (ITA) is a tuple
  $\A=\langle\Sigma, n, Q, q_0,$ $Q_f,  \lambda, X, \act, \Delta\rangle$,
  where:
\begin{itemize}
	\item $\Sigma$ is a finite alphabet; 
	\item $n$ is the number of levels;
	\item $Q$ is a finite set of states, $q_0$ is the initial
          state and $Q_f$ is a subset of $Q$ of final states.  The
          mapping $\lambda : Q \rightarrow \{1, \ldots, n\}$
          associates with each state its level. We denote by
          $Q_i=\lambda^{-1}(i)$ the set of states at level $i$;
      \item $X=\biguplus_{i=1}^n X_i$ is the set of clocks
        partitionned according to the levels and $X_i=\{x_i\}\uplus
        Y_i$ includes a \emph{main} clock $x_i$ and a set of
        \emph{auxiliary} clocks $Y_i$. The set of main clocks of
        levels less than $k$ is denoted by $X_{<k}=\{x_i\mid i<k\}$ ;
      \item $\act: Q \rightarrow X$ with $q \in Q_i \Rightarrow
        \act(q) \in X_i$ associates with a state its \emph{active}
        clock;
%
%
%
%
      \item $\Delta \subseteq Q \times \C(X) \times (\Sigma \cup
        \{\eps\}) \times \U(X) \times Q$ is a finite set of
        transitions.  Let $q \tr{\fee, a, u} q'$ be a transition in
        $\Delta$ with $k=\lambda(q)$ and $k'=\lambda(q')$. The guard
        $\fee$ is a constraint in $\C(X_k,X_{<k})$.
	\begin{itemize}
        \item if $k \leq k'$ then the update $u$ is of the
          form $$\bigwedge_{z \in \bigcup_{i\leq k}X_i} z := C_z$$
         \item if $k > k'$ then
        the update $u$ is of the form 
        $$\bigwedge_{z \in \bigcup_{i\leq k'}X_i} z := C_z\: \wedge 
        \bigwedge_{z \in \bigcup_{k'<i\leq k}X_i} z := 0$$

	\end{itemize}
            where,  when $z \in X_i$, 
            \begin{itemize}
                \item either $C_z=z$, meaning that $z$ is unchanged;
		\item or $C_z=\sum_{j<i} a_jx_j+b$, \textit{i.e.},
                  $z$ is updated by an expression over main clocks of
                  lower levels;
		\item or $C_z =z' \in X_i$ if $z \in Y_i$ or $i =
                  k=k'$, \textit{i.e.}, $z$ is updated by another
                  clock at the same level\\ under the condition that
                  $z$ is not a main clock of level lower than the
                  current one \footnote{The motivation for this rather
                    elaborate condition is explained in the
                    reachability decision procedure.}.
	   \end{itemize}	
\end{itemize}
\end{definition}

The semantics of an ITA is described by a transition system, where a
configuration $(q,v)$ consists of a state $q$ of the ITA and a clock
valuation $v$.

\begin{definition} 
\label{def:semantics}
The semantics of an ITA $\A$ is defined by the (timed) transition
system $\T_{\A}= (S, s_0, \rightarrow)$.  The set of configurations is
$S=\left\{\!(q,v) \mid q \in Q, \ v \in \R^X \!\right\}\!$, with
initial configuration $s_0=(q_0, \vect{0})$. The relation
$\rightarrow$ on $S$ consists of two types of steps:
\begin{description}
\item[Time steps:] Only the active clock in a state can evolve, all
  other clocks are suspended.  For a state $q$, 
  a time step of duration $d$ is defined by $(q,v)
  \tr{d} (q, v')$ with $v'(\act(q))=v(\act(q))+ d$ and
  $v'(x)=v(x)$ for any other clock $x$. We write $v'=v+_q d$.

\item[Discrete steps:] A discrete step $(q, v) \tr{e} (q', v')$
  can occur for some transition $e=q \tr {\fee, a, u} q'$ in
  $\Delta$ such that $v \models \fee$ and $v' = v[u]$.  
\end{description}
\end{definition}
A \emph{run} of $\A$ is a finite path in the transition system
$\T_{\A}$, which can be written as an alternating sequence of
(possibly null) time and discrete steps. A state $q \in Q$ is
\emph{reachable} from $q_0$ if there is a path in $\T_{\A}$ from
$(q_0, \vect{0})$ to $(q,v)$, for some valuation $v$.  A run with
label $d_1a_1 d_2 a_2 \ldots d_n a_n$ is \emph{accepting} if it starts
in $(q_0, \vect{0})$ and ends in $(q,v)$, for some $q \in Q_f$ and
some valuation $v$. For such a run, the timed word
$w=(a_1,d_1)(a_2,d_1+d_2) \ldots (a_n, d_1+\ldots +d_n)$ (where pairs
with $\eps$ actions are removed) is said to be \emph{accepted} by
$\A$.  The timed language of $\A$, denoted by $\La(\A)$, is the set of
timed words accepted by $\A$. The untimed language of $\A$ is
$Untime(\La(\A))$.

\medskip We now show several properties of this model linked to the
presence of auxiliary clocks.

\begin{example}[Simulation of timing policies]
  The earlier definition of ITA from~\cite{berard12} is a restriction
  of Definition~\ref{def:ita} without auxiliary clocks but with a
  policy, which can be either urgent, delayed or lazy, associated with
  each state. In a lazy state time may elapse, in an urgent state time
  may not elapse and in a delayed state time must elapse.  We show in
  Figure~\ref{fig:policy-simulation} how to model timing policies with
  a dedicated auxiliary clock per level, say $y_i$. When entering a
  state $q$ of level $i$ from a state $q'$ of level $j\geq i$, $y_i$
  is updated with the active clock of $q$. By definition, when
  entering a state $q$ of level $i$ from a state $q''$ of level $k<i$,
  $y_i$ and the active clock of $q$ are null. Thus checking whether
  time has elapsed in $q$ is equivalent to check whether
  $\act(q)>y_i$. When $q$ is a lazy state there is nothing to check. 
\end{example}

\begin{figure}[ht]
\centering
\small
\subfigure[Urgent state $q$, with $k<i \leq j$]{\label{fig:urgent}
\begin{tikzpicture}[node distance=2.5cm,auto]

\node[state, scale=0.7,minimum width={35pt}] (qi) {$q,i$};
\node[state, scale=0.7,minimum width={35pt}] (qj) [node distance=2.75cm,above left of=qi] {$q',j$};
\node[state, scale=0.7,minimum width={35pt}] (qk) [node distance=2.75cm,below left of=qi] {$q'',k$};
\node (qdest)  [node distance=3cm,right of=qi] {};
\path[->] (qi) edge node {\timedtrans{}{}{y_i = \act(q)}} (qdest);
\path[->] (qj) edge node {\timedtrans{}{}{y_i:=\act(q)}} (qi);
\path[->] (qk) edge (qi);

\end{tikzpicture}

}
\qquad \qquad
\subfigure[Delayed state $q$, with $k<i \leq j$]{\label{fig:delayed}
\begin{tikzpicture}[node distance=2.5cm,auto]

\node[state, scale=0.7,minimum width={35pt}] (qi) {$q,i$};
\node[state, scale=0.7,minimum width={35pt}] (qj) [node distance=2.75cm,above left of=qi] {$q',j$};
\node[state, scale=0.7,minimum width={35pt}] (qk) [node distance=2.75cm,below left of=qi] {$q'',k$};
\node (qdest)  [node distance=3cm,right of=qi] {};
\path[->] (qi) edge node {\timedtrans{}{}{y_i < \act(q)}} (qdest);
\path[->] (qj) edge node {\timedtrans{}{}{y_i:=\act(q)}} (qi);
\path[->] (qk) edge (qi);

\end{tikzpicture}
}

\caption{Simulating timing policies}
\label{fig:policy-simulation}
\end{figure}
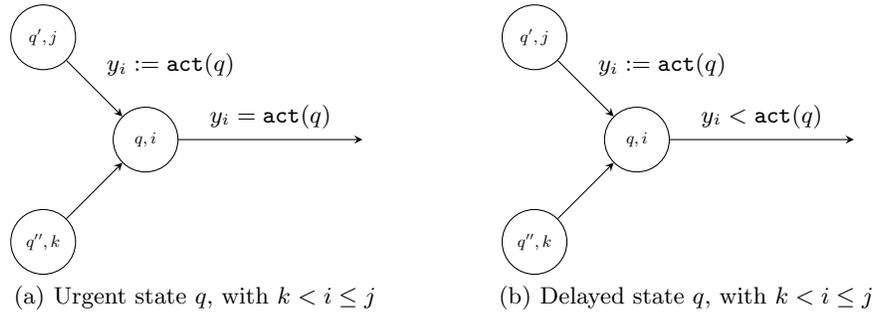


\begin{example}[About expressiveness]
  Consider the ITA $\A_1$ of Figure~\ref{fig:expressiveness} with a
  single level and single final state $q_2$. The main clock $x$ is
  active in all states and $y$ is an auxiliary clock. Its untimed
  language is $(ab)^+$. In the accepted timed words, there is an
  occurrence of $a$ at each time unit and the successive occurrences
  of $b$ come each time closer to the next occurrence of $a$ than
  previously.  More formally, its timed language $L=\La(\A_1)$ is
  defined by:
\begin{eqnarray*}
L=\big\{ (a,t_1)(b,t_2) &\ldots& (a,t_{2p+1})(b,t_{2p+2}) 
\mid p \in \N,\\
&&\forall 0\leq i\leq p,\ t_{2i+1}=i+1 \mbox{ and }  i+1<t_{2i+2}<i+2, \\
&&\forall 1\leq i\leq p, \ t_{2i+2} - t_{2i+1}< t_{2i} - t_{2i-1} \big\}
\end{eqnarray*}
It has been shown in~\cite{berard12} that this timed language cannot
be accepted by an ITA without auxiliary clocks, which yields the next 
proposition.
\end{example}

\begin{proposition}
  There exists a timed language of an ITA with a single level and one
  auxiliary clock that cannot be accepted by an ITA without
  auxiliary clocks.
\end{proposition}

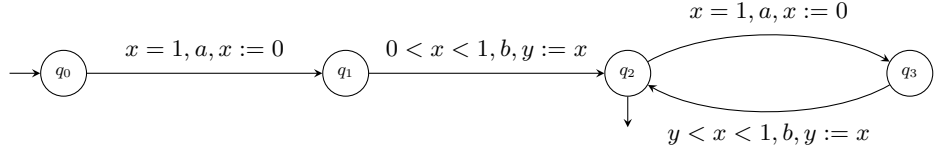
\begin{figure}[ht]
\centering
\begin{tikzpicture}[node distance=5cm,auto]
\tikzstyle{every state}+=[scale=0.75]
\node[state,initial] (q0) at (0,0) {$q_0$};
\node[state] (q1) [right of=q0] {$q_1$};
\node[state,accepting,accepting where=below] (q2) [right of=q1] {$q_2$};
\node[state] (q3) [node distance=5cm,right of=q2] {$q_3$};

\path[->] (q0) edge node {\timedtrans{}{}{x=1, a, x:=0}} (q1);
\path[->] (q1) edge node {\timedtrans{}{}{0<x<1, b, y:=x}} (q2);
\path[->] (q2) edge [bend left,looseness=0.75] node {\timedtrans{}{}{x=1, a, x:=0}} (q3);
\path[->] (q3) edge [bend left,looseness=0.75] node 
{\timedtrans{y<x<1,b,y:=x}{}{}} (q2);
\end{tikzpicture}
\caption{ITA $\A_1$ with an auxiliary clock}
\label{fig:expressiveness}
\end{figure}


Adding auxiliary clocks also has an impact on the complexity
of decision problems for ITA. In~\cite{berard12}, it is shown
that the state reachability problem is in PTIME for a fixed
number of levels without auxiliary clocks. 
The next proposition establishes a lower bound for this
problem in ITA with a single level.

\begin{proposition}\label{prop:pspacehard}
The state reachability problem for ITA with a single level
is PSPACE-hard.
\end{proposition}

\begin{proof}
  We proceed by reducing the planification problem to our reachability
  problem. The planification problem is defined by $n$ propositional
  variables $p_1,\ldots,p_n$ and a set $R$ of $m$ rules.
  Each rule $r \in R$ is defined by a guard $\bigwedge_{j=1}^{k} l_j$,
  with litterals $l_j \in \{p_1,\neg p_1,\ldots,p_n,\neg p_n\}$, and
  an update $\bigwedge_{j=1}^{h} p_{\alpha_j} := b_j$ with $b_j\in
  \{{\bf false},{\bf true}\}$. Initially all propositions are false
  and the planification problem consists in deciding whether there
  exists a sequence of rules $r_{1}\ldots r_{k}$ applicable from the
  initial state and leading to the state where all propositions are
  true.

\noindent
The corresponding ITA has $n$ auxiliary clocks $y_1,\ldots,y_n$ and
two states $q_0$ (the initial one) and $q_1$ (the final one) both with
active clock $x_1$. Each rule yields a transition looping around $q_0$
and an additional transition from $q_0$ to $q_1$ ``checking'' that the
goal has been reached. This reduction is illustrated in
Figure~\ref{fig:reduction}.
\end{proof}

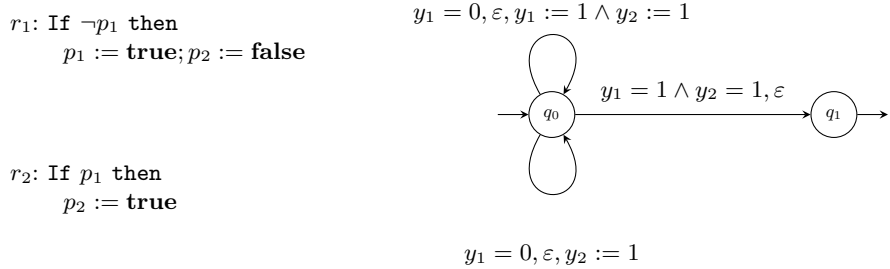
\begin{figure}[ht]
\centering
\begin{tikzpicture}[node distance=5cm,auto]
\tikzstyle{every state}+=[scale=0.75]

\node[text width=5cm] at (-4.7,1) 
{$r_1$: {\tt If} $\neg p_1$ {\tt then}\\ \hspace*{0.6cm} $p_1:={\bf true};p_2:={\bf false}$};

\node[text width=5cm] at (-4.7,-1) 
{$r_2$: {\tt If} $p_1$ {\tt then}\\ \hspace*{0.6cm} $p_2:={\bf true}$};

\node[state,initial] (q0) at (0,0) {$q_0$};
\node[state, accepting] (q1) [right of=q0] {$q_1$};

\path[->] (q0) edge node {\timedtrans{}{}{y_1=1 \wedge y_2=1, \eps}}
(q1); 
\path[->] (q0) edge [loop above, in=60,out=120,looseness=10] node[above=1pt] {$y_1=0,
    \eps, y_1:=1 \wedge  y_2:=1$} (q0); 
\path[->] (q0) edge [loop below, in=-60,out=-120,looseness=10] node[below=16pt] {$y_1=0, \eps, y_2:=1$} (q0);

\end{tikzpicture}
\caption{Illustrating the reduction for PSPACE-hardness}
\label{fig:reduction}
\end{figure}


\section{Reachability analysis of ITA}
\label{sec:reach-ita}

We prove in this section that the untimed language of an ITA is a
regular language for which a finite automaton can effectively be
built. Similarly to previous cases, the proof is based on the
construction of a (finite) class graph which is time abstract
bisimilar to the transition system $\T_{\A}$.  This result also holds
for infinite words with standard B\"uchi conditions.  As a
consequence, we obtain decidability of the reachability problem, as
well as decidability for plain ${\sf CTL}^*$ model-checking.

The construction of classes is much more involved than in the case of
TA. More precisely, it depends on the expressions occurring in the
guards and updates of the automaton (while in TA it depends only on
the maximal constant occurring in the guards). Given an ITA $\A$ with
$n$ levels, we associate with each state $q$ of $\A$ a set of
expressions $Exp(q)$ with the following meaning. The values of clocks
giving the same ordering of these expressions correspond to a class.
In order to define $Exp(q)$, we first build a family of sets
$\{E_k\}_{1\leq k\leq n}$ and set $Exp(q)= \bigcup_{k\leq \lambda(q)}
E_k$.  Finally we show in Theorem~\ref{prop:reach} how to build the
class graph which proves the regularity of the untimed language.  This
immediately yields a reachability procedure given in
Theorem~\ref{prop:reachita}.

\subsection{Construction of $\{E_k\}_{k\leq n}$}
\label{subsec:contructionexpressions}

\medskip We first recall the \emph{normalization}
operation~\cite{berard12}, on expressions relative to some level. As
explained below, this operation will be used to order expression
values at a given level.

\begin{definition}[Normalization]
  Let $k\leq n$ and $C=\sum_{i\leq k}a_ix_i+b$ be an expression over
  clocks in $X_{< k+1}$, the \emph{$k$-normalization} of $C$, denoted
  by ${\tt norm}(C,k)$, is defined by:
\begin{itemize}
	\item if $a_k\neq 0$ 
then ${\tt norm}(C,k)=x_k+(1/a_k)(\sum_{i<k}a_ix_i+b)$;
	\item else ${\tt norm}(C,k)=C$. 
\end{itemize}
\end{definition}

Let $C \rel 0$ be a guard occurring in a transition outgoing from a
state $q$ with level $k$ and $C= a_k z+\sum_{i<k}a_ix_i+b$ with $z \in
X_k$ (in the saturation procedure we do not consider guards of the
form $z-z'$ with $z,z'$ in $X_k$).  By rescaling the expression and if
necessary changing the comparison operator we may assume that $C$ is
written as $\alpha z+\sum_{i<k}a_ix_i+b$, with $\alpha \in \{0,1\}$.

\bigskip The construction of $\{E_k\}_{k\leq n}$ must be adapted to
handle auxiliary clocks. It proceeds top down from level $n$ to level
$1$ after initialization $E_k=X_k \cup \{0\}$ for all $k$.  When level
$k$ is handled, new terms are added to $E_i$ for $1\leq i\leq k$.
These expressions are those needed to compute a (pre)order on the
expressions in $E_k$.

\begin{enumerate}
\item At level $k$, first for each expression $\alpha
  z+\sum_{i<k}a_ix_i+b$ (with $\alpha \in \{0,1\}$ and $z\in X_k$)
  occurring in a guard of an edge leaving a state of level $k$, we add
  $-\sum_{i<k}a_ix_i-b$ to $E_k$.
\item Then the following procedure is iterated until no new term is
added to any $E_i$ for $1\leq i\leq k$.
	\begin{enumerate}
        \item Let $q \tr{\fee,a,u} q'$ with $\lambda(q)\geq k$ and
          $\lambda(q')\geq k$. For any $C \in E_{k}$, we add $C[u]$ to
          $E_{k}$. Observe that due to our restrictions on updates
          $C[u]$ is still either of the form $z \in X_k$ or of the
          form $\sum_{j<k}a_jx_j+b$.
      \item Let $q \tr{\fee,a,u} q'$ with $\lambda(q) < k$ and
        $\lambda(q') \geq k$.  Let $C$ and $C'$ be two different
        expressions in $E_{k}$. We compute $C''={\tt
          norm}(C[u]-C'[u],\lambda(q))$, choosing an arbitrary order
        between $C$ and $C'$ in order to avoid redundancy. Let us
        write $C''$ as $\alpha
        x_{\lambda(q)}+\sum_{i<\lambda(q)}a_ix_i+b$ with $\alpha \in
        \{0,1\}$. Then we add $-\sum_{i<\lambda(q)}a_ix_i-b$ to
        $E_{\lambda(q)}$.
	\end{enumerate}
\end{enumerate}


\begin{lemma}
\label{prop:terminate}
For an ITA $\A$, let $H$ be the number of constraints in the guards,
$U$ the number of updates in the transitions (we assume $U \geq 2$)
and $M= \textrm{max}\{ card(X_k) \mid 1 \leq k \leq n\}$. The construction
procedure of $\{E_k\}_{k\leq n}$ terminates and the size of every
$E_k$ is bounded by $(H+M)^{2^{n-k}}\times U^{2^{n(n-k+1)}}$.
\end{lemma}
\begin{proof}
  Given some $k$, we prove the termination of the stage relative to
  $k$. Observe that step 2 of the iteration only adds new expressions
  to $E_{h}$ for $h<k$. Thus steps 1 and 2 can be ordered.  Let us
  prove the termination of step 1. We define $E_{k}^{0}$ as the set
  $E_k$ at the beginning of this stage and $E_{k}^{i}$ as this set
  after insertion of the $i^{th}$ item in it. With each added item
  $C[u]$ can be associated its \emph{father} $C$. Thus we can view
  $E_{k}$ as an increasing forest with finite degree (due to the
  finiteness of the edges) and finitely many roots. Assume that this
  step does not terminate. Then we have an infinite forest and by
  K\"onig lemma, it has an infinite branch $C_0,C_1,\ldots$ where
  $C_{i+1}=C_i[u_i]$ for some update $u_i$ such that $C_{i+1}\neq
  C_i$. Observe that updates of the form $x:= x'$ do not modify the
  set. Moreover, the number of updates that change the variables $x
  \in X_k$ is either 0 or 1 since once $x$ disappears it cannot appear
  again.  We split the branch into two parts before and after this
  update or we still consider the whole branch if there is no such
  update. In these (sub)branches, we conclude with the same reasoning
  that there is at most one update that change the variables $x \in
  X_{k-1}$. Iterating this process, we conclude that the number of
  updates is at most $2^k-1$ and the length of the branch is at most
  $2^k$.

  The final size of $E_{k}$ is thus at most $E_{k}^{0}\times U^{2^k}$
  since the width of the forest is bounded by $U$.
  In step 2, we add at most $U \times (|E_{k}|\times(|E_{k}|-1))/2$ to
  $E_{i}$ for every $i<k$. This concludes the proof of termination.

  \bigskip We now prove by a backward induction that as soon as $n\geq
  2$, $|E_{k}|\leq (H+M)^{2^{n-k}}\times U^{2^{n(n-k+1)}}$.  The
  doubly exponential size of $E_n$ (proved above) is propagated
  downwards by the saturation procedure. We define $p_k = |E_{k}|$.

\paragraph{Basis case $k=n$.}
We have $p_n \leq p_n^{0}\times U^{2^n}$ where $p_n^{0}$ is bounded by
$H+M$, hence $p_n \leq (H+M)\times U^{2^n}$ which is the claimed
bound.

\paragraph{Inductive case.}
Assume that the bound holds for $k < j \leq n$.  Due to all executions
of step 2 of the procedure at strictly higher levels, $p_k^0$
expressions were added to $E_k$, with:
\begin{eqnarray*}
p_k^{0} &\leq& (H+M) + U \times \left[(p_{k+1}\times (p_{k+1}-1))/2 + \cdots + (p_{n}\times (p_{n}-1))/2\right]\\
p_k^{0} &\leq& (H+M) + U \times \left[ (H+M)^{2^{n-k}}U^{2^{n(n-k)+1}} + \cdots + (H+M)^2 U^{2^{n +1}}\right]\\
p_k^{0} &\leq& (n-k+1) \times (H+M)^{2^{n-k}}U^{2^{n(n-k)+1}} \quad\textrm{(replacing all terms by the largest) }\\
p_k^{0} &\leq& (H+M)^{2^{n-k}}\times U^{2^{n(n-k+1)+n}} \quad\textrm{(here we use } U \geq 2 \textrm{ and } n \geq 2)
\end{eqnarray*}
Taking into account step 1 of the procedure for level $k$, we
have: \[p_k \leq (H+M)^{2^{n-k}}\times U^{2^{n(n-k)+1}+2^k+n}.\] Let
us consider the term $\delta = 2^{n(n-k+1)}-2^{n(n-k)+1}-2^k-n =
2^{n(n-k)+1}(2^{n-1}-1) - 2^k -n$. We have $\delta \geq 2^{n+1} - 2^n \geq 0$,
which yields the claimed bound.
\end{proof}

In order to analyze the space requirements triggered by the saturation
procedure, we establish the following lemma bounding the number of
bits used for integers involved in the rational constants of
expressions in all $E_k$.

\begin{lemma}
\label{lemma:size}
Let $\A$ be an ITA, and let $b_0$ be the maximal number of bits for
integers occurring in
$\A$. If $b$  is the number of bits of an integer constant, 
occurring in an expression of some $E_k$, then $b \leq ((n+1)!)^29^{n}b_0$.
\end{lemma}
\begin{proof}
 Without loss of generality we assume that $b_0\geq 2$.
 We also assume that there is a single denominator $s$ for the rationals
  occurring in updates since it only induces a polynomial blow up.

 \noindent Let $b_{k}$ be the number of bits of an integer occurring
 in some expression before operations of level $n-k$ are performed.
 We establish a relation between $b_k$ and $b_{k+1}$.  At level $n-k$,
 step 1 involves a normalization on guards. Thus a numerator is
 multiplied by a denominator to produce the new integers leading to a
 number of bits $2b_k$.
 For an expression that was already present in $E_{n-k}$, 
 its coefficients are modified in order to get a common denominator
 by taking the product of the original denominators.
  After this transformation the maximal number of bits
  is bounded by $(n-k+1)b_k$.
  
  \noindent Let $C=\sum_{i\leq {n-k}}a_ix_i+b$ be an expression built
  after step 2(a).  Examining the successive updates, the coefficient
  $a_i$ can be expressed as $\sum_{d\in \mathcal D} \prod_{j \in d}
  c_{d,j}$ where $\mathcal D$ is the set of subsets of
  $\{i,\ldots,n-k\}$ containing $i$ and $c_{d,j}$ are either
  coefficients of the updates or coefficients of an expression built
  before this step. The same reasoning applies to $b$.  Before summing
  the products over $d \in \mathcal D$, the integers are transformed
  in order to get the same denominator by multiplying every
  denominator (and corresponding numerator) by $s^i$ with $0\leq i
  \leq n-k$.  So the maximal absolute value of the numerator of such a
  coefficient is bounded by
  $2^{n-k}(2^{(n-k+1)b_k})^{n-k+1}2^{(n-k)b_0} \leq
  (2^{2b_k+1})^{(n-k+1)^2}$ which implies a maximal number of bits
  equal to $(n-k+1)^2(2b_k+1)$ for the numerators of the $a_i$'s and
  $b$. The maximal absolute value of the denominator of such a
  coefficient is less than $(2^{(n-k+1)b_k})^{n-k+1}2^{(n-k)b_0}$
  which implies a maximal number of bits bounded by $(n-k+1)^2(2b_k)$
  for the denominators of the $a_i$'s and $b$.

\noindent At step 2(b), the difference $C[u]-C'[u]$ requires to
compute the lcm of two denominators (bounded by their product). So the
difference operation leads to a bound $(n-k+1)^2(4b_k+2)$ for the
numerators of its coefficients and  $(n-k+1)^2(4b_k)$ for the denominators.

\noindent The final step 2(b) consists in multiplying a numerator and
a denominator of some coefficients leading to a bound
$(n-k+1)^2(8b_k+2)\leq (n-k+1)^2(9b_k)$ for $b_{k+1}$, which yields the
desired bound.

\end{proof}

\subsection{Construction of the class automaton}
\label{subsec:contructiongraph}

In order to analyze the size of the class automaton defined below, we
recall an adaptation of a classical result about partitions of
$n$-dimensional Euclidian spaces.

\begin{definition}
  Let $\{H_k\}_{1\leq k \leq m}$ be a family of hyperplanes of
  $\R^n$. A \emph{region} defined by this family is a connected
  component of $\R^n \setminus \bigcup_{1\leq k \leq m} H_k$.  An
  \emph{extended region} defined by this family is a connected
  component of $\bigcap_{k \in I} H_k \setminus \bigcup_{k \notin I}
  H_k$ where $I \subseteq \{1,\ldots, m\}$ with the convention that
  $\bigcap_{k \in \emptyset} H_k=\R^n$.
\end{definition}

\begin{proposition}\label{prop:zas+}~\\
1. \cite{zas75}
  The number of regions defined by the family $\{H_k\}_{1\leq k \leq
    m}$ is at most $\sum_{i=0}^n \binom{m}{i}$.\\
2. \cite{berard12} The number of extended regions
  defined by the family $\{H_k\}_{1\leq k \leq m}$ is at most:\\
  $\sum_{p=0}^n\binom{m}{p}\sum_{i=0}^{n-p} \binom{m-p}{i}\leq
  e^2m^n$.
\end{proposition}


\begin{theorem}
\label{prop:reach}
The untimed language of an ITA is regular.
\end{theorem}

\begin{proof}
  Starting from an ITA $\A$, and handling auxiliary clocks, we build a
  finite automaton which is time abstract bisimilar to the transition
  system $\T_{\A}$ and thus accepts $Untime(\La(\A))$.
\paragraph{Class definition.}
A state of the automaton, called class, is a syntactical
representation of a subset of reachable configurations. It is defined
as a pair $R=(q,\{\preceq_k\}_{1 \leq k \leq \lambda(q)})$ where $q$
is a state and $\preceq_k$ is a total preorder over $E_k$, for $1 \leq
k \leq \lambda(q)$.  The class $R$ describes the set of
configurations:
\[\sem{R}=\{(q,v) \mid  \forall k \leq \lambda(q)\ 
\forall (g,h) \in E_k,\ g[v] \leq h[v] \mbox{~iff~} g \preceq_k h\}\]

The initial state is the class $R_0$ such that $\sem{R_0}$ contains
$(q_0,{\bf 0})$ and can be straightforwardly determined. The final
states are all classes $R=\left(q,\{\preceq_k\}_{1 \leq k \leq
    \lambda(q)}\right)$ with $q \in Q_f$.



Observe that fixing a state, the set of configurations $\sem{R}$ of a
non empty class $R$ is exactly an extended region associated with the
hyperplanes defined by the comparison of two expressions of some
$E_k$. An upper bound for the total number of expressions of any level
is given by $(H+M)^{2^n} \times U^{2^{n^2}}$, hence an upper bound of
the of the number of hyperplanes is obtained by squaring this number,
yielding $(H+M)^{2^{n+1}} \times U^{2^{n^2}}$. Using Point 2. of
Proposition~\ref{prop:zas+}, the number of semantically different
classes for a given state is bounded by:
\begin{equation}
e^2m^n=e^2(H+M)^{K2^{n+1}} \times U ^{K 2^{n^2+1}}
\label{eq:number}
\end{equation}
where 
$K = \sum_{k=1}^n card(X_k) \leq nM$ is the total number of clocks.
Since semantical equality between classes can be tested in polynomial
time w.r.t. their size~\cite{RoTeVi97}, we implicitely consider in the
sequel of the proof classes modulo the semantical equivalence.

There are two kinds of transitions, corresponding to discrete steps
and abstract time steps.

\paragraph{Discrete step.}
Let $R=(q,\{\preceq_k\}_{1 \leq k \leq \lambda(q)})$ and
$R'=(q',\{\preceq'_k\}_{1 \leq k \leq \lambda(q')})$ be two classes. There is a
transition $R \tr{e} R'$ for a transition $e : q \tr{\fee,a,u} q'$ if
there is some $(q,v) \in\, \sem{R}$ and $(q',v') \in\, \sem{R'}$ such that
$(q,v) \tr{e} (q',v')$. In this case, for all $(q,v) \in\, \sem{R}$
there is a $(q',v') \in\, \sem{R'}$ such that $(q,v) \tr{e}
(q',v')$. This can be decided as follows.

\emph{Firability condition.} For a transition $e$ like above at level
$\ell = \lambda(q)$, write $\fee=\bigwedge_{j \in J} C_j \bowtie_j 0$.
Since we assumed rescaled guards, for every $j$, $C_j=\alpha
z+\sum_{i<k}a_ix_i+b$ (with $\alpha \in \{0,1\}$ and $z$ in $X_\ell$)
or $C_j=z-z'$ with $z,z' \in X_\ell$.  In the first case
$C'_j=-\sum_{i<\ell}a_ix_i-b$ and $z$ belong to $E_\ell$ and in the
second case $z,z' \in E_\ell$ both by construction. For each $j \in J$,
we define a condition depending on $\bowtie_j$. For instance, in the
first case if the constraint in $\fee$ is $C_j \leq 0$, we check that
$\alpha z \preceq_\ell C'_j$, or if the constraint in $\fee$ is $C_j > 0$
we check that $\alpha z \npreceq_\ell C_j' \wedge C_j' \preceq_\ell \alpha
z$. The second case is handled similarly.

\smallskip \emph{Successor definition.}
Class $R'$ is defined as follows. Let $k \leq \lambda(q')$ and $g,h \in E_k$.
\begin{enumerate}
\item Either $k \leq \ell$, then by construction, $g[u],h[u] \in
E_k$ then $g \preceq'_k h$ iff $g[u] \preceq_k h[u]$.
\item Or $k > \ell$, let $D=g[u]-h[u]$. Due to our restrictions on
  updates for $i\leq \ell$, $x_i[u]$ can only be equal to $x_i$ or
  $\sum_{j<i} \alpha_jx_j+\beta$.  Thus $D$ can be written as
  $\sum_{i\leq\ell}c_ix_i+d$. We set $C={\tt norm}(D,\ell)$ and write
  $C=\alpha x_{\ell}+\sum_{i<\ell}a_ix_i+b$ (with $\alpha \in
  \{0,1\}$).
  By construction, $C'=-\sum_{i<\ell}a_ix_i-b \in E_{\ell}$.\\
  When $c_{\ell}\geq 0$ then $g \preceq'_k h$ iff $\alpha
  x_{\ell} \preceq_{\ell} C' $.\\
  When $c_{\ell}< 0$ then $g \preceq'_k h$ iff $C'
  \preceq_{\ell} \alpha x_{\ell}$.

\end{enumerate}
By definition of $\sem{\,\cdot\,}$, we obtain:
\begin{itemize}
\item For any $(q,v) \in \sem{R}$, if there exists $(q,v) \tr{e}
(q',v')$ then the firability condition is fulfilled and $(q',v')$
belongs to $\sem{R'}$.
\item If the firability condition is fulfilled then for each $(q,v) \in
\sem{R}$ there exists $(q',v') \in \;\sem{R'}$ such that $(q,v)
\tr{e} (q',v')$.
\end{itemize}

\paragraph{Time step.}
Let $R=(q,\{\preceq_k\}_{1 \leq k \leq \lambda(q)})$, with again $\ell
= \lambda(q)$.  There is a transition $R \tr{succ} Post(R)$ for
$Post(R)=(q,\{\preceq'_k\}_{1 \leq k \leq \ell})$, the time
successor of $R$, which is defined as follows.

For every $i < \ell$, we define $\preceq'_i=\preceq_i$.  Let $\sim$ be
the equivalence relation $\preceq_{\ell} \cap \preceq^{-1}_{\ell}$
induced by the preorder. On equivalence classes, this (total) preorder
becomes a (total) order.  Let $V$ be the equivalence class containing
$\act(q)$.
\begin{enumerate}
\item Either $V=\{\act(q)\}$ and it is the greatest equivalence
  class. Then $\preceq'_{\ell}=\preceq_{\ell}$ (thus
  $Post(R)=R$).
\item Either $V=\{\act(q)\}$ and it is not the greatest equivalence
  class.  Let $V'$ be the next equivalence class. Then
  $\preceq'_{\ell}$ is obtained by merging $V$ and $V'$, and
  preserving $\preceq_{\ell}$ elsewhere.
\item Either $V$ is not a singleton. Then we split $V$ into
  $V\setminus \{\act(q)\}$ and $\{\act(q)\}$ and ``extend''
  $\preceq_{\ell}$ by $V \setminus \{\act(q)\}
  \preceq'_{\ell} \{\act(q)\}$.
\end{enumerate}
By definition of $\sem{\,\cdot\,}$, for each $(q,v) \in \sem{R}$,
there exists $d>0$ such that $(q,v+d) \in \sem{Post(R)}$ and for each $d$
with $0 \leq d' \leq d$, then $(q,v+d') \in \sem{R} \cup \sem{Post(R)}$.

\smallskip From the properties above, this finite automaton accepts
$Untime(\La(\A))$.


\end{proof}

\begin{theorem}
\label{prop:reachita}
The reachability problem for Interrupt Timed Automata is decidable and
belongs to 2-EXPTIME. It is in PTIME when the number of
clocks is fixed and PSPACE-complete when the number of levels is
fixed.
\end{theorem}
\begin{proof}
  The reachability problem is solved by building the class graph and
  applying a standard reachability algorithm.  The number of
  expressions in the $E_k$'s is doubly exponential w.r.t the size of
  the model (see Lemma~\ref{prop:terminate}). The size of an
  expression is exponential w.r.t. the size of the model (see
  Lemma~\ref{lemma:size}).  So the size of a class representation is
  also doubly exponential in the size of the model.  The size of the
  graph, bounded by the number of semantically different classes, is
  only polynomial w.r.t. the size of a class due to Point 2. of
  Proposition~\ref{prop:zas+}.  This leads to a 2-EXPTIME complexity.
  Observe that no complexity gain can be obtained by a non
  deterministic search without building the graph.
  
  \noindent
  Again using these lemmas and Point 2. of
  Proposition~\ref{prop:zas+}, when the number of clocks is fixed the
  size of the graph is at most polynomial in the size of the
  problem, leading to a PTIME procedure.
 
  \noindent
  On the other hand, when the number of levels is fixed, the size of a
  class representation is polynomial while the number of classes is
  exponential (see $K$ in Equation~(\ref{eq:number})).  Thus a non
  deterministic search can be performed without building the graph,
  which yields a complexity in PSPACE. The PSPACE hardness is a
  consequence of Proposition~\ref{prop:pspacehard}.
\end{proof}

\noindent \textbf{Remarks.} This result should be compared with the
similar one for TA. The reachability problem for TA is PSPACE-complete
and thus less costly to solve than for ITA. 
Fixing the number of levels in ITA yields the same complexity.
Moreover, fixing the
number of clocks does not reduce the complexity for TA (when this
number is greater than or equal to $3$) while this problem belongs now
to PTIME for ITA. Summarizing, the main source of complexity for ITA is
the number of levels and clocks, while in TA it is the binary encoding of the
constants~\cite{courcoubetis92}.

\section{Parametric Interrupt Timed Automata}
\label{sec:pmodele}

Parametric ITA are similar to ITA but they include polynomials
of parameters from a set $P$, in guards and updates.  Given two sets
$F,G$, we denote by $\Po(F,G)$, the set of polynomials with variables
in $F$ and coefficients in $G$ and by $\Fr(F,G)$, the set of rational
functions with variables in $F$ and coefficients in $G$
(i.e. quotients of polynomials). Observe that $\Li(F,G)$ can be seen
as the subset of polynomials with degree at most one.

\begin{definition}
  A \emph{parametric interrupt timed automaton} (PITA) is a tuple
  $\A=\langle P, \Sigma, n, Q, q_0, Q_f,  \lambda, X,$ $\act, \Delta\rangle$,
  where:
\begin{itemize}
\item $P$ is a finite set of parameters, 
\item all other elements are defined as for ITA except that
  expressions appearing in guards or updates belong to
  $\Li(X,\Po(P,\Q))$: in such an expression $\sum_{z \in Z} a_z z+b$, the
  $a_z$'s and $b$ are polynomials over $P$ with coefficients in $\Q$.
%
\end{itemize}
\end{definition}

This definition implies that an ITA is a PITA with $P =
\emptyset$. When all expressions occurring in guards and updates are
in $\Li(X\cup P,\Q)$ (which can be seen as a subset of
$\Li(X,\Po(P,\Q))$), the PITA is said to be \emph{additively
  parametrised}. In contrast, in the general case, it is called
\emph{multiplicatively parametrised}.

\smallskip As in the unparametrized case, updates operate on
expressions.  For instance, for clocks in $X = \{x_1, x_2\}$,
parameters in $P=\{p_1,p_2,p_3\}$, expression $C= p_2x_2 -2x_1 + 3p_1$
and the update $u$ defined by $x_1 := 1 \wedge x_2 := p_3x_1 +p_2$,
applying $u$ to $C$ yields the expression $C[u] = p_2p_3x_1 + p_2^2 +
3p_1 -2$. Note that the use of multiplicative parameters for clocks
may result in polynomial coefficients when updates are applied.  Here
a \emph{clock valuation} is a mapping $v : X \mapsto \Po(P,\R)$.  For
a valuation $v$ and an expression $C \in \Li(X,\Po(P,\Q))$, $v(C) \in
\Po(P,\R)$ is obtained by evaluating $C$ w.r.t. $v$.  Given an update
$u$ and a valuation $v$, the valuation $v[u]$ is defined by $v[u](x) =
v(C_x)$ for $x$ in $X$ if $x := C_x$ is the update for $x$ in $u$ and
$v[u](x) =v(x)$ otherwise.  For instance, let $X = \{x_1, x_2, x_3\}$
be a set of three clocks. For valuation $v = (2p_2, 1.5, 3p_1^2)$ and
update $u$ defined by $x_1 := 1 \wedge x_3 := p_1x_3 - x_1$, applying
$u$ to $v$ yields the valuation $v[u] = (1, 1.5, 3p_1^3-2p_2)$.

A \emph{parameter valuation} is a mapping $\pi : P \mapsto \R$.  
For a parameter valuation $\pi$ and an expression $C \in \Li(X,\Po(P,\Q))$,
$\pi(C) \in \Li(X,\R)$ is obtained by evaluating $C$ w.r.t. $\pi$.  If
$C \in \Po(P,\Q)$, then $\pi(C)\in \R$.  Given a parameter valuation
$\pi$, a clock valuation $v$ and an expression $C \in
\Li(X,\Po(P,\Q))$ we write $\pi,v \models C \rel 0$ when $\pi(v(C))
\rel 0$. 

\smallskip
Given a parameter valuation $\pi$ and a PITA $\A$, substituting the
parameters by their value according to $\pi$ yields an ITA, denoted by
$\A(\pi)$, where the coefficients of clocks are in $\R$.  So the
semantics of $\A$ w.r.t. parameter valuation $\pi$ is defined by the
(timed) transition system $\T_{\A(\pi)}$. A state $q$ is reachable
from $q_0$ for valuation $\pi$ if $q$ is reachable from $q_0$ in
$\A(\pi)$.


\begin{example}
  A PITA $\A_2$ is depicted in \figurename~\ref{fig:exita1}, with two
  interrupt levels. Every level $i$ has only a main clock $x_i$.
  Fixing the parameter valuation $\pi$: $p_1=5$ and
  $p_2 = -1$, the run $(q_1,0,0)\xrightarrow{4}
  (q_1,4,0)\xrightarrow{a}
  (q_2,4,0)\xrightarrow{3}(q_1,4,2)\xrightarrow{b} (q_2,4,3)$ is
  obtained as follows. After staying in $q_1$ for $4$ time units, $a$
  can be fired and the value of $x_1$ is then frozen in state $q_2$,
  while $x_2$ increases. Transition $b$ can be taken if
  $x_1+p_2x_2=2$, hence for $x_2 = 2$, after which $x_2$ is updated to
  $x_2=(p_1-4p_2^2)4+p_2 = 3$. A geometric view of this run
  w.r.t. $\pi$ is given (in bold) in \figurename~\ref{fig:traj}.
\end{example}


\begin{figure}[ht]
\centering
\small
\subfigure[A PITA $\A_2$ with two interrupt levels]{\label{fig:exita1}
\begin{tikzpicture}[node distance=2.5cm,auto]

\node[state,initial,initial text={},scale=0.7] (q0) at (0,0) {$q_1,1$};
\node[state, accepting, accepting where=above,scale=0.7] (q1) [node distance=2.75cm,above right of=q0] {$q_2,2$};
\path[->] (q0) edge node {\timedtrans{}{}{x_1 < p_1,a}} (q1);
\path[->] (q1) edge [loop right] node [xshift=-4mm] (tr) {\timedtrans{x_1 + p_2 x_2 = 2}{b}{x_2:=(p_1-4p_2^{2})x_1+p_2}} (q2);

\end{tikzpicture}

}\hfill{~}
\subfigure[A possible run in $\A_2$ for $\pi$]{\label{fig:traj}
\begin{tikzpicture}[scale=1.5]
\path[draw=black,->] (-0.3,0) -- (2.1,0) node[anchor=west] {$x_1$};
\path[draw=black,->] (0,-0.3) -- (0,1.75) node[anchor=south east] {$x_2$};
\path[draw=black,dashed] (0,0.4) -- (1.8,0.4) {};
\path[draw=black,dashed] (0,0.6) -- (1.8,0.6) {};

\node[anchor=north] (al) at (0.8,0) {$4$};
\node[anchor=east] at (0,0.4) {$2$};
\node[anchor=east] at (0,0.6) {$3$};
\node[anchor=north] at (1,0) {$5$};

\node[anchor=north] at (1,2) {$x_1=p_1$};
\node[anchor=north] at (2.2,1.2) {$x_1+p_2x_2=2$};

\draw (1,1.75) -- (1,0) node [above] (TextNode) {};

\draw (0,-0.4)--(1.8,1.4) node [midway, below] (TextNode) {};

\node[anchor=north] (d) at (0,0) {};
\node[anchor=north] (a) at (0.8,0) {};
\node[anchor=south west] (b) at (0.8,0.4) {};
\node[anchor=south west] (c) at (0.8,0.6) {};
\path[draw=black,very thick] (d.north) -- (a.north);
\path[draw=black,very thick] (a.north) -- (b.south west);
\path[draw=black,very thick] (b.south west) -- (c.south west);
\fill(d.north) circle (0.05);
\fill(a.north) circle (0.05);
\fill (b.south west) circle (0.05);
\fill (c.south west) circle (0.05);

\end{tikzpicture}
}

\caption{An example of PITA and a possible execution}
\label{fig:exita1traj}
\end{figure}
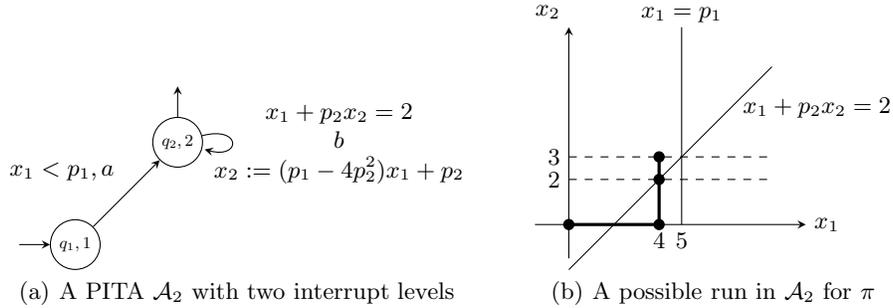

\paragraph{Reachability problems.} We consider several reachability
problems for this class. Let $\A$ be a PITA with initial state $q_0$
and $q$ be a state of $\A$.  The \emph{Existential (resp. Universal)
  Reachability Problem} asks whether $q$ is reachable from $q_0$ for
some (resp. all) parameter valuation(s).  \emph{Scoped} variants of
these problems are obtained by adding as input a set of parameter
valuations given by a first order formula over the reals or a
polyhedral constraint.  The \emph{Robust Reachability Problem} asks
whether there exists a parameter valuation $\pi$ and a real $\eps > 0$
such that for all $\pi'$ with $\|\pi - \pi'\|_{\infty}< \eps$, $q$ is
reachable from $q_0$ for $\pi'$ (where $\| \pi \|_{\infty} = max_{p\in
  P} |\pi(p)|$). When satisfied, this property ensures that small
parameter perturbations do not modify the reachability result. It is
also related to parameter synthesis where a valuation has to be
enlarged to an open region with the same reachability goal.


\section{Reachability Analysis with Additive Parametrization}
\label{sec:amodele}

We start with the easier particular case of additive parametrization,
\textit{i.e.}, expressions occurring in guards and updates are linear
expressions on clocks and parameters with rational coefficients.  We
first prove that the existential parametrized reachability problem can
be reduced to the reachability problem on (non-param\-etrized) ITA.

\begin{proposition}
  For any additively parametrized PITA $\mathcal{A}$, with set of
  states $Q$ and initial state $q_0$, there exists a (non-parametrised)
  ITA $\mathcal{A}'$, with set of states $Q'$, containing $Q$, and initial
  state $q_0'$ fulfilling the following equivalence.
  For every $q\in Q$:\\
  \centerline{there exists $\pi$ such that $q$ is reachable from $q_0$
    in $\mathcal A$ for $\pi$} \centerline{iff $q$ is reachable from
    $q_0'$ in $\A'$.}

\label{p1}
\end{proposition}

\begin{proof}
  For any additively parametrized PITA $\mathcal{A}$ with $n$ levels,
  and $k$ parameters $p_1,...,p_k$, we build an equivalent ITA
  $\mathcal{A}'$ with $n+k+1$ levels and then use the complexity
  results of section~\ref{sec:reach-ita}. The construction is shown in
  \figurename~\ref{easyPITA}.

  \noindent
  The ITA $\mathcal{A}'$ consists of a ``prefix'' (the first $k+1$
  levels) connected to the original automaton $\mathcal{A}$ (with its
  $n$ levels). The main clocks of levels $1$ to $k$ encode the
  parameters $p_1,\ldots,p_k$ of $\A$.  In order to simplify further
  references, we also call these clocks $p_1,...,p_k$.  Similarly, the
  main clock of the first level is called $p_0$. None of these $k+1$
  first level has any auxiliary clock.  Since level numbers start at
  $1$, each clock $p_i$ is active in level $i+1$ in (the prefix of)
  $\mathcal{A}'$.

  \noindent
  In the first level of $\mathcal{A}'$, clock $p_0$ is active. After some
  arbitrary time, a transition, with no guard, is taken to the state of the
  second level and clock $p_0$ is frozen. In the second level, clock $p_1$ is
  active and the same procedure continues: after some time a transition to the
  next level is taken, and clock $p_1$ is frozen, and so on for the first $k$
  levels. 
  In these first $k$ levels, we any run of $\mathcal A'$
  choses a non-negative fixed value for the
  clocks $p_0,\ldots,p_{k-1}$, and hence almost for the parameters of
  $\mathcal{A}$. Parameters may however have negative values so level $k+1$
  serves as a technicality to choose the final sign of the corresponding
  clocks.  This is done by assigning $p_{i-1}$ or $-p_{i-1}$ to clock $p_i$,
  between each two consecutive states, for all $i\in[1..k-1]$, in a run without
  any delay in any of the states of level $k+1$ (the other runs, with delays in
  the states of level $k+1$, overlap on those corresponding to other parameter
  valuations and are therefore not a problem). In the last state of level
  $k+1$, the frozen clocks $p_1,...,p_k$ can therefore have any arbitrary real
  value assigned. 
  The automaton finally proceeds to the initial state of $\mathcal{A}$ keeping
  the values of these additional clocks. Since they correspond to levels lower
  than any level of $\mathcal{A}$ they can be used liberally enough in the
  guards and updates of $\mathcal{A}$.  The obtained automaton $\mathcal{A}'$
  is an ITA and parameters of $\mathcal{A}$ are modeled as clocks in
  $\mathcal{A}'$.

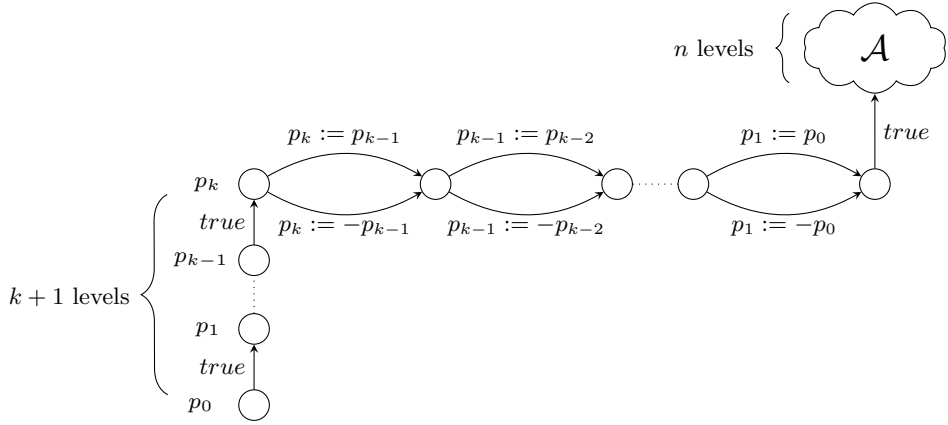
\begin{figure}
\begin{tikzpicture}[auto,scale=0.2]

\draw node[state,scale=0.5] (q0) {};
\draw node[state, above =0.6 of q0, scale=0.5] (q1) {};
\draw node[state, above=0.5 of  q1,scale=0.5] (q3) {};
\draw node[state, above=0.6 of q3,scale=0.5] (q4) {};

\draw node[state, right=2 of q4,scale=0.5] (q5) {};
\draw node[state, right=2 of q5,scale=0.5] (q6) {};
\draw node[state, right=0.6 of q6,scale=0.5] (q7) {};
\draw node[state, right=2 of q7,scale=0.5] (q8) {};

\node (q11)[left=4mm,scale=1.2, font=\scriptsize] at (q0) {$p_0$};
\node (q111)[left=3mm,scale=1.2] at (q11) {};

\node [left=3mm,scale=1.2,font=\scriptsize] at (q1) {$p_1$};
\node [left=2mm,scale=1.2, font=\scriptsize] at (q3) {$p_{k-1}$};
\node (q10)[left=3mm,scale=1.2,font=\scriptsize] at (q4) {$p_{k}$};
\node (q100)[left=4mm,scale=1.2] at (q10) {};
\node (q9)[cloud,cloud puffs=10,cloud puff arc=120, aspect=2,draw=black,above=12mm,scale=1.5] at (q8) {$\mathcal{A}$};

\node (q1111)[yshift=6mm,left=10mm,scale=1.2] at (q9) {};
\node (q0000)[yshift=-6mm,left=10mm,scale=1.2] at (q9) {};
\draw [decorate,decoration={brace,amplitude=5pt}] (q0000)--(q1111) node [xshift=-10mm,yshift=-6mm] {\small $n$ levels};

\draw[->] (q0)to node {$true$}(q1);
\draw[->] (q3) to node {$true$}(q4);
\draw[dotted] (q1) to node {} (q3);
\draw[->](q8) to node [ swap] {$true$} (q9);

\draw [decorate,decoration={brace,amplitude=10pt}] (q111)--(q100) node [xshift=-13mm,yshift=-15mm] {\small $k+1$ levels};

\draw[->] (q4) to [bend left] node {$p_k:=p_{k-1}$} (q5);
\draw[->] (q4) to [bend right] node[yshift=-4mm] {$p_k:=-p_{k-1}$} (q5);

\draw[->] (q5) to [bend left] node {$p_{k-1}:=p_{k-2}$}(q6);
\draw[->] (q5) to [bend right] node [yshift=-4mm]{$p_{k-1}:=-p_{k-2}$} (q6);

\draw[->] (q7) to [bend left] node {$p_1:=p_0$}(q8);
\draw[->] (q7) to [bend right] node [yshift=-4mm]{$p_1:=-p_0$} (q8);

\draw[dotted] (q6)--(q7);

\end{tikzpicture}
\caption{An equivalent ITA $\mathcal{A}'$}
\label{easyPITA}
\end{figure}

  \noindent
  Let $X$ be the set of clocks in $\mathcal{A}$ and $X'$ be the set of
  clocks in $\mathcal{A}'$ (thus $X'=X\cup \{p_0,...,p_k\}$). For any
  subset $Y\subseteq X$ and a valuation $v$, we define the restriction
  of $v$ to $Y$ as the unique valuation $v$ on $Y$ such that
  $v_{|Y}(x)=v(x)$. We now show that a configuration $s=(q,v)$ is
  reachable in $\mathcal{A}$ for some parameter valuation $\pi$ (i.e.,
  in $\mathcal{A}(\pi)$) iff there exists some configuration
  $s'=(q',v')$, such that $q'=q$ and for all $x\in X,
  v'_{|X}(x)=v(x)$, is reachable in $\mathcal{A}'$.

  \noindent
On the one hand, if there exists a path to reach $s'$ in $\mathcal{A}'$, then
by construction this path goes through a configuration $(q_0,v_0)$ such that
$(q_0,v_{0|X})$ is the initial configuration of $\mathcal{A}$ (i.e.
$v_{0|X}$ is the zero valuation). Let $\pi$ be the parameter valuation
such that for all $i>0, \pi(p_i)=v_{0}(p_i)$, then $s$ is reachable in
$\mathcal{A}({\pi})$. 

  \noindent
  On the other hand, let $\pi$ be a parameter valuation and $v$ be a
  clock valuation on $X$ such that $(q,v)$ is reachable in
  $\mathcal{A}(\pi)$.  Then using an appropriate run in the prefix one
  reaches $(q_0,v_0)$ with $v_{0|X}$ is the zero valuation and for all
  $i>0, v_{0}(p_i)=\pi(p_i)$. Afterwards this run is extended to reach
  $q$ by mimicking the run of $\mathcal{A}(\pi)$.

\end{proof}

Using Proposition~\ref{p1} and Theorem~\ref{prop:reachita}, we can now give
the main result of this section.

\begin{theorem}\label{thm:reach-add}
The (polyhedral scoped) existential reachability problem is
  decidable for additively pa\-rametrised PITA, and belongs to 2-EXPTIME.
  It belongs to PTIME when the number of clocks and parameters is
  fixed. It is PSPACE-complete when the number of levels and parameters is
  fixed.
\end{theorem}

\begin{proof}
  Following Proposition \ref{p1}, every additively parametrised PITA
  can be transformed into an equivalent ITA, and the (unscoped)
  reachability problem of additively parametrised PITA is thus reduced
  to the reachability problem of ITA, already known to be
  decidable. The complexity results follow from the complexity results
  for ITA given in Theorem~\ref{prop:reachita}, since the size of
  $\mathcal{A}'$ is only linear in the size of $\mathcal{A}$: if there
  are $n$ levels, $N$ clocks, $k$ parameters, $x$ states and $y$
  transitions in $\mathcal{A}$, the number of levels, clocks, states
  and transitions in $\mathcal{A}'$ are $n+k+1$, $N+k+1$, $x+2k+1$ and
  $y+3k+1$, respectively.

  \noindent
With a polyhedral scope, given as a finite union of polyhedra, we need to guard
the transition between the last state of the prefix and the initial state of
$\mathcal{A}$, in $\mathcal{A}'$, by the given polyhedra (each polyhedra of the
union could guard a different transition, as well).

\end{proof}

\section{Reachability Analysis with Multiplicative Parametrization}
\label{sec:multpar}
We now focus on the multiplicative case and this section is devoted to
the proof of the following result:
\begin{theorem}\label{thm:reach}
  The (scoped) existential, universal and robust reachability problems
  for PITA are decidable and belong to 2-EXPSPACE. The complexity
  reduces to PSPACE when the number of levels is fixed.
\end{theorem}

We first present the main ideas underlying the proof, which is based
on the proof of Theorem~\ref{prop:reachita} but extends it by the
handling of parameters.  Given a PITA $\A$, the first step is to build
a finite partition of the set $\R^P$ of parameter valuations. An
element $\Pi$ of this partition is specified by a satisfiable
first-order formula over $(\R,+,\times)$, with the parameters as
variables. Intuitively, inside $\Pi$ the qualitative behaviour of $\A$
does not depend on the precise parameter valuation.  In a second step,
we build a finite automaton $\Rl(\Pi)$ for each non empty $\Pi$.  In
$\Rl(\Pi)$, a state $R$, again called a class, defines a set
$\sem{R}_{\pi}$ of reachable configurations of $\T_{\A(\pi)}$ for a
valuation $\pi \in \Pi$. The transition relation of $\Rl(\Pi)$
contains discrete steps $R \tr{e} R'$ (for a transition $e$ of $\A$)
and abstract time steps $R \tr{} Post(R)$ with the following
properties:
\begin{description}
\item[Discrete Step (DS):] If there is a transition $R \tr{e} R'$ in
  $\Rl(\Pi)$ then for each $\pi \in \Pi$ and each $(q,v) \in
  \sem{R}_{\pi}$ there exists $(q',v') \in \;\sem{R'}_{\pi}$ such that
  $(q,v) \tr{e} (q',v')$.

  Conversely, let $\pi \in \Pi$ and $(q,v) \in \sem{R}_{\pi}$. If there
  exists a transition $(q,v) \tr{e} (q',v')$ in $\T_{\A(\pi)}$ then for
  some $R'$, there is a transition $R \tr{e} R'$ in
  $\Rl(\Pi)$ and $(q',v')$ belongs to $\sem{R'}_{\pi}$.
\item[Time Step (TS):] Let $\pi \in \Pi$ and $(q,v) \in
  \sem{R}_{\pi}$. There exists $d>0$ such that $(q,v+_qd) \in
  \sem{Post(R)}_{\pi}$ and for each $d'$ with $0 \leq d' \leq d$,
  $(q,v+_q d') \in \sem{R}_{\pi} \cup \sem{Post(R)}_{\pi}$.
\end{description}
Hence, we obtain a finite family of abstract time bisimulations of the
transition systems $\T_{\A(\pi)}$, for all parameter valuations, which
gives the decidability result.

\smallskip Although the construction of $\Rl(\Pi)$ is similar to the
one for ITA, expressions in the sets $\{E_k\}_{k\leq n}$ now contain
polynomials of parameters. The main difference is the normalization
operation of an expression $\sum_{i\leq k} a_ix_i+b$ which depends on
the polynomial $a_k$.  For instance, consider expression
$p_2x_2+x_1-2$ which appear in automaton $\A_2$ of
\figurename~\ref{fig:exita1} with a comparison to $0$. For a valuation
where $p_2=0$, a normalization should yield $x_1-2$.  If $p_2 \neq 0$,
the operation should yield $-\frac{x_1-2}{p_2}$.  In addition, the
case $p_2\neq 0$ should be split depending on the sign of $p_2$, since
the operation could change the comparison operator involved in a
guard. Therefore, we also need to define a set $PolPar$ of polynomials
appearing in the denominators like $p_2$.

\subsection{Construction of $PolPar$  and expressions $\{E_k\}_{k\leq n}$}
\label{subsec:pcontructionexpressions}
In the spirit of normalization, we define three operations on
expressions, relatively to a level $k$, to help building the elements
in $E_k$ to which the active clock on level $k$ will be compared.

\begin{definition}
  Let $k\leq n$ be some level and let $C$ be an
  expression in $\Li(X_{<n+1},$ $\Fr(P,\Q))$, 
   $C=\sum_{i\leq n}a_ix_i+b$ with $a_k = \frac{r_k}{s_k}$, for
  some $r_k$ and $s_k$ in $\Po(P, \Q)$. We associate with $C$ the
  following expressions:
\begin{itemize}
	\item ${\tt lead}(C,k)=r_k$;
	\item if ${\tt lead}(C,k) \notin \Q\setminus\{0\}$, ${\tt
            comp}(C,k)=\sum_{i< k}a_ix_i+b$;
	\item if ${\tt lead}(C,k)\neq 0$ then 
	${\tt compnorm}(C,k)=-\sum_{i<k}\frac{a_i}{a_k}x_i-\frac{b}{a_k}$.
\end{itemize}
\end{definition}

In the previous example, {\tt comp} corresponds to $x_1-2$ while {\tt
  compnorm} corresponds to $-\frac{x_1-2}{p_2}$. More examples are
given after the construction of $PolPar$ and $\{E_k\}_{k\leq n}$. This
construction proceeds top down from level $n$ to level $1$ after
initialising $PolPar$ to $\emptyset$ and $E_k$ to $X_k\cup\{0\}$ for all
$k$. When handling level $k$, we add new terms to $E_i$ for $1\leq
i\leq k$.

\begin{enumerate}
\item At level $k$ the first step consists in adding new expressions
  to $E_k$ and new polynomials to $PolPar$. More precisely, let $C$ be
  any expression occurring in a guard of an edge leaving a state of
  level $k$.  We add  
  ${\tt lead}(C,k)$ to $PolPar$ when it does not belong to $\Q$ and
  we add ${\tt comp}(C,k)$ and ${\tt compnorm}(C,k)$ to $E_k$
   when they are defined.
\item The second step consists in iterating the following procedure
  until no new term is added to any $E_i$ for $1\leq i\leq k$.
	\begin{enumerate}
        \item Let $q \tr{\fee,a,u} q'$ with $\lambda(q)\geq k$ and
          $\lambda(q')\geq k$, and let $C \in E_{k}$. Then we add $C[u]$
          to $E_{k}$.
        \item Let $q \tr{\fee,a,u} q'$ with $\lambda(q) < k$ and
          $\lambda(q') \geq k$.  Let $\{C,C'\}$ be a set of two
          expressions in $E_{k}$. We compute $C''=C[u]-C'[u]$,
          choosing an arbitrary order between $C$ and $C'$. 
          This step ends by handling $C''$
          w.r.t. $\lambda(q)$ as done for $C$ w.r.t. $k$ in step 1
          above.

%
	\end{enumerate}
\end{enumerate}

\begin{example}
  For the automaton of \figurename~\ref{fig:exita1}, initially, we
  have $PolPar=\emptyset$, $E_1=\{x_1, 0\}$ and $E_2=\{x_2,
  0\}$. Starting with level $k=2$, we consider in step 1 the
  expression $C_2=p_2x_2+x_1-2$ appearing in the guard of the single
  edge leaving $q_2$. We compute ${\tt lead}(C_2,2)=p_2$, ${\tt
    comp}(C_2,2)=x_1-2$, and ${\tt
    compnorm}(C_2,2)=-\frac{x_1-2}{p_2}$. We obtain $PolPar=\{p_2\}$
  and $E_2=\{x_2, 0, x_1-2, -\frac{x_1-2}{p_2}\}$.
  For step 2(a) and the same edge, we apply its update
  to the expressions of $E_2$ that contain $x_2$, add them to $E_2$,
  and thus obtain $E_2=\{x_2, 0, x_1-2,
  -\frac{x_1-2}{p_2},(p_1-4p_2^2)x_1+p_2\}$.

  In step 2(b), considering the single edge from $q_1$ to $q_2$,
  we compute the differences between any two expressions from $E_2$
  (after applying update which means here substituting 0 to $x_2$
  and letting $x_1$ unchanged) and the resulting expressions ${\tt lead}$,
  ${\tt comp}$ and ${\tt compnorm}$, which yields:\\
$PolPar=\{p_2, p_2+1,1-p_1+4p_2^2,1+p_1p_2-4p_2^3\}$,\\ $E_1=\{x_1, 0, 2,
  -\frac{2(p_2+1)}{p_2},-2-p_2, \frac{2+p_2}{1-p_1+4p_2^2},
  \frac{p_2^2-2}{p_2}, \frac{2-p_2^2}{1+p_1p_2-4p_2^3}
  \}$.

  We proceed with level $1$.
Since expression $C_1=x_1-p_1$ occurring in the guard of the considered edge
has leading coefficient equal to 1,
there is no term to add to $PolPar$.
We add ${\tt
    compnorm}(C_1,1)=p_1$ to $E_1$, hence the final result is:\\
 $PolPar=\{p_2, p_2+1,1-p_1+4p_2^2,1+p_1p_2-4p_2^3\}$,\\ 
 $E_1=\{x_1, 0, 2,
  -\frac{2(p_2+1)}{p_2},-2-p_2, \frac{2+p_2}{1-p_1+4p_2^2},
  \frac{p_2^2-2}{p_2}, \frac{2-p_2^2}{1+p_1p_2-4p_2^3},p_1
  \}$,\\
   $E_2=\{x_2, 0, x_1-2,
  -\frac{x_1-2}{p_2},(p_1-4p_2^2)x_1+p_2\}$.

\end{example}

Lemma~\ref{prop:polpar} below is used for the class automata
construction.  Its proof is obtained by a straightforward examination
of the above procedure. 
%
%

\begin{lemma}
\label{prop:polpar}
Let $C$ belong to $E_k$ for some $k$ and $c=\frac{r}{s}$ be a
coefficient of $C$ with $s \notin \Q$. Then there exists polynomials
$P_1,\ldots, P_\ell \in PolPar$ and some constant $K \in
\Q\setminus\{0\}$ such that $s = K. \prod_{1 \leq i \leq \ell} P_i$.
\end{lemma}

Lemma~\ref{prop:pterminate} is the parametrized version of
Lemma~\ref{prop:terminate} and its (omitted) proof is almost
identical.

\begin{lemma}
\label{prop:pterminate}
For a PITA $\A$, let $H$ be the number of constraints in the guards,
$U$ the number of updates in the transitions (we assume $U \geq 2$)
and $M= \textrm{max}\{ card(X_k) \mid 1 \leq k \leq n\}$. The construction
procedure of $\{E_k\}_{k\leq n}$ terminates and the size of every
$E_k$ is bounded by $(H+M)^{2^{n-k}}\times U^{2^{n(n-k+1)}}$.
\end{lemma}

Lemma~\ref{lemma:psize} is the parametrized version of
Lemma~\ref{lemma:size}.  However since the coefficients are now
rational functions, the degree of the polynomials must also be
analyzed.

\begin{lemma}
\label{lemma:psize}
Let $\A$ be a PITA, and let $b_0$ be the maximal total number of bits
for integers of an expression in $\A$ and $d_0$ the maximal degree of
polynomials, occurring in $\A$. If $b$ is the total number of bits of
the integer constants and $d$ the degree of a polynomial, occurring in
an expression of $PolPar$ or some $E_k$, then $b \leq ((n+1)!)^2
(n+1)2^{3n+1}b_0$ and $d \leq (n+1)!5^nd_0$.
\end{lemma}
\begin{proof}
  W.l.o.g. we assume that there is a single denominator for the rationals
  occurring in updates since it only induces a polynomial blow up.

 \noindent
 Assume that before the level $n-k$ is performed, the total number of
 bits for integers occurring in some expression is $b_{k}$.  We
 establish by induction that $b_k\leq \prod_{j=1}^k (n+2-j)^2
 (k+1)2^{n+2k+1}b_0$.  The basis case is trivial.  At level $n-k$,
 step 1 does induces an increasing only when operation {\tt compnorm}
 is applied on a original guard whose coefficients are polynomials
 (instead of rational fractions).  After this operation the number of
 bits is bounded by $(n-k+1)b_0\leq (n-k+1)b_k$.  For an expression
 that was already present in $E_{n-k}$, its coefficients are modified
 in order to get a common denominator by taking the product of the
 original denominators.  After this transformation the total number of
 bits is bounded by $(n-k+1)2b_k$.
%

  \noindent  Examining one
  update applied on an expression, the total number
  of bits of the coefficients of the updated expression
  is increased by $(n-k+1)b_0$.
%
  Since an expression built after step 2(a) has been obtained
  by less than $2^{n-k}$ updates, the total number of
  bits is less than 
  $(n-k+1)2b_k+2^{n-k}(n-k+1)b_0$.
 
\noindent At step 2(b), the difference $C[u]-C'[u]$ requires to
compute the lcm of two denominators (bounded by their product). So the
difference operation leads to a bound $(n-k+1)4b_k+2^{n-k+1}(n-k+1)b_0$
for the total number of bits.

\noindent The final step 2(b) consists in multiplying a numerator and
a denominator of some coefficients leading to a bound:\\
$(n-k+1)^2(4b_k+2^{n-k+1}b_0)$\\
$\leq 
(n-k+1)^2\left(4\prod_{j=1}^k (n+2-j)^2 
(k+1)2^{n+2k+1} b_0+2^{n-k+1}b_0\right)$\\
$\leq \left(\prod_{j=1}^{k+1} (n+2-j) \right)^2  
((k+1)2^{n+2(k+1)+1}b_0+2^{n+2(k+1)+1}b_0)$\\
$= \left(\prod_{j=1}^{k+1} (n+2-j) \right)^2  
(k+2)(2^{n+2(k+1)+1}b_0)$\\
for the number of bits.

\smallskip \noindent Assume that before the level $n-k$ is performed,
the degree of a polynomial (of parameters) occurring in some
expression is at most $d_{k}$.  We establish a relation between $d_k$
and $d_{k+1}$.  At level $n-k$, step 1 does not induce any increasing
when operation 
{\tt compnorm} is applied on a original guard
whose coefficients are polynomials (instead of rational fractions).  
More precisely the
numerators of rational fractions are unchanged while the denominators
are numerators of some previous expressions.
  For an expression that was already present in $E_{n-k}$, 
  its coefficient are modified in order to get a common denominator
  by taking the product of the original denominators.
  After this transformation the maximal degree
  is bounded by $(n-k+1)d_k$.

\noindent Let us examine an expression $C=\sum_{i\leq {n-k}}a_ix_i+b$
built after step 2(a).  Examining the successive updates, the
numerator of coefficient $a_i$ can be expressed as $\sum_{d\in
  \mathcal D} \prod_{j \in d} c_{d,j}$ where $\mathcal D$ is the set
of subsets of $\{i,\ldots,n-k\}$ containing $i$ and $c_{d,j}$ 
are all coefficients of the updates (i.e. coefficients of polynomials)
 except one coefficient of
the expression built before this step. 
The same reasoning applies to
the constant coefficient of the expression. 
 So the degree of the $a_i$'s and $b$ is bounded by:
$(n-k+1)(d_k+d_0)$. The denominators are denominators of expressions
previously built so bounded by $(n-k+1)d_k$.
 
\noindent At step 2(b), the difference $C[u]-C'[u]$ requires to
compute the lcm of two denominators (bounded by their product). So the
difference operation leads to a bound $(n-k+1)(2d_k+d_0)$ for the numerators of
its coefficients and $(n-k+1)2d_k$ for the denominators.

\noindent The final step 2(b) consists in multiplying a numerator and
a denominator of some coefficients leading to a bound $(n-k+1)(4d_k+d_0)$. So
$d_{k+1}\leq (n-k+1)5d_k$ yielding the desired bound.
\end{proof}

We now explain the partition construction. Starting from the finite
set $PolPar$, we split the set of parameter valuations in parameter
regions specified by the result of comparisons to $0$ of the values of
the polynomials in $PolPar$. For instance, for the set $PolPar$
computed above, the inequalities $p_2 < 0$, $p_2+1=0$,
$1-p_2-4p_1^2=0$ and $1+p_1p_2-4p_2^3=0$ define a set $preg$ of
parameter valuations. The parameter region $preg$ is non empty since
it contains $p_1=5$ and $p_2=-1$.  The set of such constraints
yielding non empty regions can be computed by solving an existential
formula of the first-order theory of reals.

Then, given a non empty parameter region $preg$, we
consider the following subset of $E_k$ for $1 \leq k \leq n$:
$\ E_{k,preg} = \{ C \in E_k \ \mid$ the denominators
of coefficients of $C$ are non null in  $preg \}$.
Due to Lemma~\ref{prop:polpar}, these subsets are obtained 
by examining the specification of $preg$.

Observe that expressions in $E_{1,preg} \setminus X_1$ belong to
$\Fr(P,\Q)$ and that, depending on the parameter valuation, 
the values of two expressions can be differently ordered. We refine $preg$
according to a linear pre-order $\preceq_1$ on $E_{1,preg}\setminus
X_1$ which is satisfiable within $preg$. We denote this refined
region by $\Pi=(preg,\preceq_1)$ and we now build a finite automaton
$\Rl(\Pi)$.

\subsection{Construction of the class automata}
\label{subsec:pcontructiongraph}
In this paragraph, we fix a non empty parameter region
$\Pi=(preg,\preceq_1)$.

\smallskip \noindent {\bf Class definition.}  A state of $\Rl(\Pi)$,
called a class like before, is defined as a pair $R=(q,\{\preceq_k\}_{1
  \leq k \leq \lambda(q)})$ where $q$ is a state of $\A$ and
$\preceq_k$ is a total preorder over $E_{k,preg}$, for $1 \leq k \leq
\lambda(q)$.  For a parameter valuation $\pi \in \Pi$, the class $R$
describes the following subset of configurations in $\T_{\A,\pi}$:

\centerline{$\sem{R}_{\pi}=\{(q,v) \mid \; \forall k \leq \lambda(q)\
  \forall g,h \in E_{k,preg}, \ \pi(v(g)) \leq \pi(v(h)) \mbox{~iff~}
  g \preceq_k h\}$}

The initial state of $\Rl(\Pi)$ is the class $R_0$, such that
$(q_0,{\bf 0}) \in \sem{R_0}_{\pi}$, which can be straightforwardly
determined by extending $\preceq_1$ to $E_{1,preg}$ with $x\preceq_{1}
0$ and $0\preceq_{1} x$ for all $x\in X_1$, and closing $\preceq_{1}$
by transitivity.

Transitions in $\Rl(\Pi)$ consist of the following discrete and time steps:

\paragraph{Discrete step.}
Let $R=(q,\{\preceq_i\}_{1 \leq i \leq \lambda(q)})$ and
$R'=(q',\{\preceq'_i\}_{1 \leq i \leq \lambda(q')})$ be two classes
and let $e : q \tr{\fee,a,u} q'$ be a transition in $\A$. There is a
transition $R \tr{e} R'$ if for some $\pi\in \Pi$, there are some
$(q,v) \in\, \sem{R}_{\pi}$ and $(q',v') \in\, \sem{R'}_{\pi}$ such
that $(q,v) \tr{e} (q',v')$. In this case, we claim that for all
$(q,v) \in\, \sem{R}_{\pi}$ there is a $(q',v') \in\, \sem{R'}_{\pi}$
such that $(q,v) \tr{e} (q',v')$. For this, we prove in the sequel
that the existence of transition $R \tr{e} R'$ is independent of $\pi
\in \Pi$ and of $(q,v) \in\, \sem{R}_{\pi}$. It can be seen as
follows.

\noindent We note $\ell = \lambda(q)$ for the level of transition $e$. 

\smallskip \noindent \emph{Firability condition.}  We write
$\fee=\bigwedge_{j \in J} C_j \bowtie_j 0$ with, for each $j$, either
$C_j=a_{\ell} z+\sum_{i<\ell}a_ix_i+b$ (with $z\in X_\ell$) or
$C_j=z-z'$ with $z,z' \in X_\ell$.
%
We consider three subcases of the first case.

\noindent $\bullet$ {\bf Subcase $a_{\ell}=0$.} Then 
$C_j = {\tt  comp}(C_j,\ell)\in E_{\ell,preg}$ and using the
positions of $0$ and $C_j$ w.r.t. $\preceq_{\ell}$, we can
decide whether $C_j \bowtie_j 0$.

\noindent $\bullet$ {\bf Subcase $a_{\ell}\in \Q\setminus\{0\}$.}
Then ${\tt compnorm}(C_j,{\ell})\in E_{{\ell},preg}$, hence using the
sign of $a_{\ell}$ and the positions of $z$ and ${\tt
  compnorm}(C_j,{\ell})$ w.r.t. $\preceq_{\ell}$, we can decide
whether $C_j \bowtie_j 0$.

\noindent $\bullet$ {\bf Subcase $a_{\ell}\notin \Q$.}  According
to the specification of $preg$, we know the sign of $a_{\ell}$
as it belongs to $PolPar$. In case $a_{\ell}=0$, we decide as in
the first subcase.  Otherwise, we decide as in the second subcase.

\noindent
The second case $C_j=z-z'$ is handled similarly.


\smallskip \noindent \emph{Successor definition.}
To build the successor $R'=(q',\{\preceq'_i\}_{1 \leq i \leq
  \lambda(q')})$ of $R$, we have to define the preorders
$\{\preceq'_i\}_{1 \leq i \leq \lambda(q')}$. Let $k \leq \lambda(q')$
and $g,h \in E_{k,preg}$.
\begin{enumerate}
\item Either $k \leq \ell$, by step 2(a) of the construction,
  $g[u],h[u] \in E_{k,preg}$. Then $g \preceq'_{k} h$ iff $g[u]
  \preceq_{k} h[u]$.
\item Or ${k} > \ell$, let
  $D=g[u]-h[u]=\sum_{i\leq\ell}a_ix_i+b$. There are again three subcases.

  \noindent $\bullet$ {\bf Subcase $a_{\ell}=0$.} Then $D ={\tt
    comp}(D,\ell)\in E_{\ell,preg}$, so we can decide whether $D
  \preceq_{\ell} 0$ and $g' \preceq'_{k} h'$ iff $D \preceq_{\ell} 0$.

  \noindent $\bullet$ {\bf Subase $a_{\ell}\in \Q\setminus\{0\}$.}
  Then ${\tt compnorm}(D,{\ell})\in E_{{\ell},preg}$.  There are four
  possibilities to consider. For instance if $a_{\ell}>0$ and
  $x_{\ell}\preceq_{\ell} {\tt compnorm}(D,{\ell})$ then $g'
  \preceq'_{k} h'$. The other cases are similar.

  \noindent $\bullet$ {\bf Subcase $a_{\ell}\notin \Q$.}  Let us write
  $a_{\ell}=\frac{r_{\ell}}{s_{\ell}}$. According to the specification
  of $preg$, we know the sign of $a_{\ell}$ since $r_{\ell}$ belongs
  to $PolPar$ and $s_{\ell}$ is a product of items in $PolPar$.  In
  case $a_{\ell}=0$, we decide $g' \preceq'_{k} h'$ as in the first
  case.  Otherwise, we decide in a similar way as in the second case.
  For instance if $a_{\ell}>0$ and $x_{\ell}\preceq_{\ell} {\tt
    compnorm}(D,{\ell})$ then $g' \preceq'_{k} h'$.

\end{enumerate}



\paragraph{Time step.}
For $R=(q,\{\preceq_k\}_{1 \leq k \leq \ell})$, there is a transition
$R \tr{succ} Post(R)$, where $Post(R)=(q,\{\preceq'_k\}_{1 \leq k \leq
  \ell})$ is the time successor of $R$, defined as
follows. Intuitively, all preorders below $\ell=\lambda(q)$ are fixed,
so $\preceq'_i=\preceq_i$ for each $i < \ell$. On level $\ell$, the
value of the active clock simply progresses along the one dimensional
time line, where the expressions are ordered. More precisely, let
$\sim$ be the equivalence relation $\preceq_{\ell} \cap
\preceq^{-1}_{\ell}$ induced by the preorder. A $\sim$-equivalence
class groups expressions yielding the same value, and on these
classes, the (total) preorder becomes a (total) order. Let $V$ be the
$\sim$-equivalence class containing $\act(q)$.
\begin{enumerate}
\item Either $V=\left\{\act(q)\right\}$. If $V$ is the greatest
  $\sim$-equivalence class, then $\preceq'_{\ell}=\preceq_{\ell}$ (and
  $Post(R)=R$). Otherwise, let $V'$ be the next $\sim$-equivalence
  class. Then $\preceq'_{\ell}$ is obtained by merging
  $V=\left\{\act(q)\right\}$ and $V'$, and preserving $\preceq_{\ell}$
  elsewhere.
\item Or $V$ is not a singleton. Then we split $V$ into $V\setminus
  \left\{\act(q)\right\}$ and $\left\{\act(q)\right\}$ and ``extend''
  $\preceq_{\ell}$ by $V \setminus \left\{\act(q)\right\}
  \preceq'_{\ell} \left\{\act(q)\right\}$.
\end{enumerate}
To conclude, observe that the automaton $\Rl(\Pi)$
defined above has the properties \textbf{(DS)} and \textbf{(TS)}
mentionned previously, and is hence a finite time abstract bisimulation of
$\T_{\A,\pi}$, for all parameter valuations $\pi \in \Pi$.

We are now in position to prove Theorem~\ref{thm:reach}.

\begin{proof}
  Starting from a PITA $\A$, we use the above construction, whose
  termination is guaranteed by lemma~\ref{prop:pterminate}, to design
  a non deterministic procedure for existential reachability of a
  given state $q$:
\begin{enumerate}
 \item Build $PolPar$ and $\{E_k\}_{1\leq k\leq n}$.
 \item Guess a parameter region $(preg,\preceq_1)$.
 \item Check non emptiness of $(preg,\preceq_1)$.
 \item Build the class automaton $\mathcal R(preg,\preceq_1)$
 and check whether $q$ occurs in some class. 
\end{enumerate}
For universal reachability of $q$, in step 4, one checks whether $q$
does not occur in any class. This gives us a non deterministic
procedure for the complementary problem.  For robust reachability in
step 2, one guesses an open parameter region \textit{i.e.}, only
specified by strict inequalities.

\smallskip \noindent We now analyse the complexity of these
procedures.  Due to lemmas~\ref{prop:pterminate}
and~\ref{lemma:psize}, the first step is performed in 2-EXPTIME and in
PTIME when the number of clocks is fixed. Guessing a parameter region
has the same complexity.

\noindent
The satisfiability problem for a first-order formula is in
PSPACE~\cite{canny88}. Due to lemma~\ref{prop:pterminate}, the number
$s$ of (in)equalities specifying the region fulfills
$s=O((H+M)^{2^{n}}\times U^{2^{n^2}})$ with the previous notations.
Let $b$ be the total number of bits of the integers occurring in a
constraint of the specification of the region. Due to
lemma~\ref{lemma:psize}, $b \leq ((n+1)!)^2 (n+1)2^{3n+1}b_0$. Let $d$
be the maximal degree of the polynomials occurring in the
specification of the region. Due to the same lemma, $d \leq
(n+1)!5^nd_0$.  So the emptiness problem for a region is decided in
2-EXPSPACE which becomes PSPACE when the number of levels is fixed.

\noindent
Observe now that the class automaton $\mathcal R(preg,\preceq_1)$ is
isomorphic to the class automaton of the ITA $\A(\pi)$ that would be
obtained from $\A$ with any parameter valuation $\pi$ in $\Pi=
(preg,\preceq_1)$. It has been proved in Section~\ref{sec:reach-ita}
that this automaton can be built in polynomial time w.r.t. the size of
the representation of any class. As the size of the representation of
a class of a PITA has the same order as the one of the corresponding
ITA (dominated by the doubly exponential number of expressions) and
the construction algorithms perform similar operations, this yields a
complexity of 2-EXPTIME and PSPACE when the number of levels is fixed.

\noindent
So the dominating factor of this non deterministic procedure
is the emptiness check done in 2-EXPSPACE. By Savitch theorem
this procedure can be determinised with the same complexity.

\end{proof}

%


\begin{example}
  The construction of $\mathcal R(\Pi)$ is illustrated on the
  automaton $\A_2$ from \figurename~\ref{fig:exita1}, for the region
  $\Pi=(preg,\preceq_1)$, where $preg$ was defined above by: $p_2 <
  0$, $p_2+1=0$, $1-p_2-4p_1^2=0$ and $1+p_1p_2-4p_2^3=0$. For
  $\preceq_1$, we first remove from $E_1$ the expressions with null
  denominator: $E_{1,preg} = \{x_1, 0, 2,
  -\frac{2(p_2+1)}{p_2},-2-p_2, \frac{p_2^2-2}{p_2},p_1\}$ and we
  consider the ordering on $E_{1,preg} \setminus \{x_1\}$ specified by
  the line below.


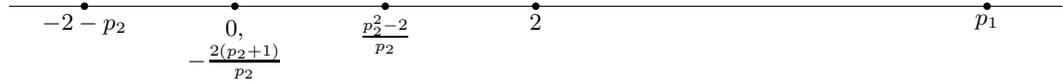
\begin{center}
\begin{tikzpicture}[auto]
\small
\draw (0,0) -- coordinate (x axis mid) (14,0);


\node[anchor=north] (g) at (1,0) {$-2-p_2$};
\fill(g.north) circle (0.05);

\node[anchor=north] (a) at (3,0) {$\begin{array}{c} 0,\\
 -\frac{2(p_2+1)}{p_2} \end{array}$};
\fill(a.north) circle (0.05);

\node[anchor=north] (f) at (5,0) {$\frac{p_2^2-2}{p_2}$};
\fill(f.north) circle (0.05);

\node[anchor=north] (d) at (7,0) {$2$};
\fill(d.north) circle (0.05);

\node[anchor=north] (c) at (13,0) {$p_1$};
\fill(c.north) circle (0.05);

\end{tikzpicture}
\end{center}

A part of the resulting class automaton $\mathcal R(\Pi)$, including
the run corresponding to the one in \figurename~\ref{fig:traj}, is
depicted in \figurename~\ref{fig:regaut}, where dashed lines indicate
(abstract) time steps.

The initial class is $R_0=(q_0, Z_0)$ where $Z_0$ is $\preceq_1$
extended with $x_1=0$. Denoting (slightly abusively) extensions with
the symbol $\wedge$, the time successors of the initial state are
obtained by moving $x_1$ to the right along the line:
$R_0^1=(q_0,\preceq_1 \wedge \ 0<x_1< \frac{p_2^2-2}{p_2})$,
$R_0^2=(q_0, \preceq_1 \wedge \ x_1= \frac{p_2^2-2}{p_2}), \ldots,$
up to
$R_0^7=(q_0,\preceq_1 \wedge \ x_1>p_1)$.  Transition $a$ can be fired
from all classes up to $R_0^5$ (but not from $R_0^6$ and $R_0^7$ where
the constraint $x_1 < p_1$ is not satisfied). In
\figurename~\ref{fig:regaut}, we represent only the one from $R_0^5
=(q_0, Z_1)$ with $Z_1 = \preceq_1 \wedge \ 2 <x_1<p_1$, corresponding
to the run in \figurename~\ref{fig:traj}.

\smallskip Along this run, the ordering $\preceq_2$ is determined by
region $\Pi$ and $Z_1$, on $E_{2,preg} \setminus\{x_2\} = \{0, x_1-2,
-\frac{x_1-2}{p_2},(p_1-4p_2^2)x_1+p_2\}$. It is illustrated on the
line below.

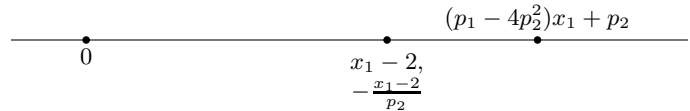
\begin{center}
\begin{tikzpicture}[auto]
\small
\draw (0,0) -- coordinate (x axis mid) (9,0);

\node[anchor=north] (a) at (1,0) {$0$};
\fill(a.north) circle (0.05);

\node[anchor=north] (e) at (5,0) {$\begin{array}{c} x_1-2, \\ 
-\frac{x_1-2}{p_2} \end{array}$};
\fill(e.north) circle (0.05);

\node[anchor=south] (f) at (7,0) {$(p_1-4p_2^2)x_1+p_2$};
\fill(f.south) circle (0.05);


\end{tikzpicture}
\end{center}
 
Firing transition $a$ produces the class $R_1=(q_1, Z_1, \preceq_2
\wedge x_2=0))$.  Transition $b$ is then fired from the (second) time
successor of $R_1$ for which $x_2=-\frac{x_1-2}{p_2}$.

\begin{figure}[ht]
\centering
\begin{tikzpicture}[node distance=3.75cm,auto]
\tikzstyle{every state}=[inner sep=2pt,draw=black,shape=rectangle,rounded corners=5pt]
\tikzstyle{time step}=[draw,->,dashed]
\tikzstyle{time succ}=[node distance=1cm]

\node[state, initial,initial text={}] (r0) at (0,0) {$R_0$};
\node[state] (r01) [node distance=1.6cm,below of=r0] {$R_0^1$}; 
\node(r02) [time succ, node distance=1.4cm, below of= r01]
{\raisebox{4pt}{$\vdots$}}; 
\node[state] (r03) [time succ,below of=r02] {$R_0^5$}; 
\node (r04) [time succ, node distance=1.2cm, below of=r03]
{\raisebox{4pt}{$\vdots$}}; \node[state] (r05) [time succ, below of
=r04] {$R_0^7$};

\node[state] (r1) [node distance=4.5cm,right of=r03] {$\begin{array}{cc}
    q_1, Z_1,\preceq_2~\wedge \\ x_2 = 0 \end{array}$};

\node[state] (r2) [above=0.7 of r1] {$\begin{array}{cc} q_1, Z_1,\preceq_2~\wedge \\  0 < x_2 < -\frac{x_1-2}{p_2} \end{array}$};

\node[state] (r3) [above =0.7 of r2] {$\begin{array}{cc} q_1, Z_1,\preceq_2~\wedge \\   x_2 = -\frac{x_1-2}{p_2} \end{array}$};

\node[state] (r18) [right=2 of r3] {$\begin{array}{cc} q_1, Z_1,\preceq_2~\wedge \\  x_2= (p_1-4p_2^2)x_1+p_2\end{array}$};

\node[state] (r19) [below=0.8 of r18] {$\begin{array}{cc} q_1, Z_1,\preceq_2~\wedge \\  x_2 > (p_1-4p_2^2)x_1+p_2 \end{array}$};



\node(d)[node distance=1.5cm,right of=r0] {$\cdots$};
\path[->] (r0) edge node {$a$} (d);

\node(d1)[node distance=1.5cm,right of=r01] {$\cdots$};
\path[->] (r01) edge node {$a$} (d1);

\coordinate[above=0.5 of r3] (d2);
\path[time step] (r3) -- (d2);

\path[time step] (r18) -- (r19);
\path[time step] (r2) -- (r3);

\path[->] (r3) edge node {$b$} (r18);
\path[->] (r03) edge node {$a$} (r1);

\path[time step] (r01) -- (r02);
\path[time step] (r03) -- (r04);
\path[time step] (r02) -- (r03);
\path[time step] (r04) -- (r05);
\path[time step] (r0) -- (r01);
\path[time step] (r1) -- (r2);

\path[time step] (r19) edge [loop below, in=-60,out=-120, looseness=4] node (tr) {} (r19);

\end{tikzpicture}
\caption{A part of $\mathcal R(\Pi)$ for $\A_2$}
\label{fig:regaut}
\end{figure}
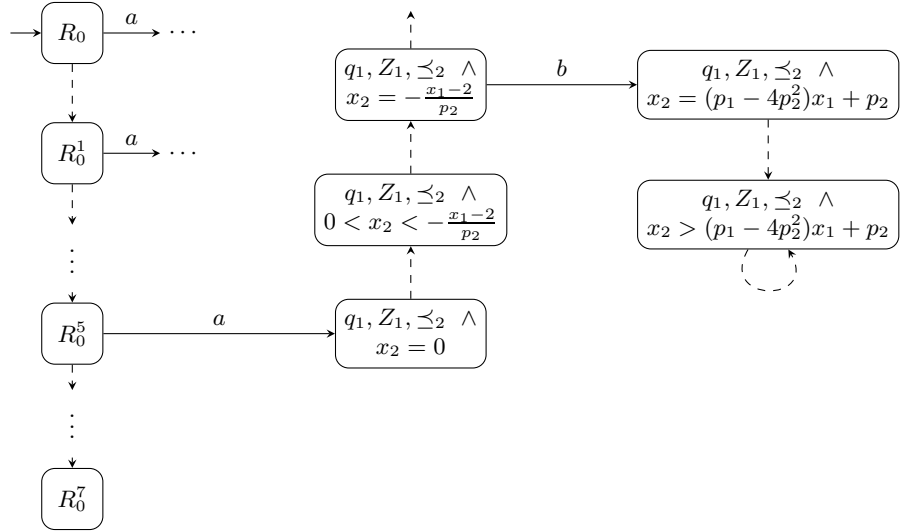

%
%
%
\end{example}


\section{Conclusion}
\label{sec:conc}

While seminal results on parametrised timed models leave little hope for
decidability in the general case, we provide here an expressive formalism
for the analysis of parametric reachability problems. Our setting includes
a restricted form of stopwatches and polynomials in the parameters
occurring as both additive and multiplicative coefficients of the clocks
in guards and updates. We plan to investigate which kind of timed temporal
logic would be decidable on PITA.



\end{document}